\pdfoutput=1

\documentclass[10pt]{article}

\usepackage{url}
\usepackage{mathtools}
\usepackage{amssymb}
\usepackage{amsthm}
\usepackage{latexsym}
\usepackage[shortlabels]{enumitem}
\usepackage{dsfont}
\usepackage{appendix}
\usepackage{color} 
\usepackage[utf8]{inputenc}
\usepackage[T1]{fontenc}
\usepackage{geometry}
\usepackage{todonotes}
\usepackage{lmodern}
\usepackage{anyfontsize}
\usepackage{stmaryrd}
\usepackage{comment}
\usepackage{bm}
\usepackage{mathrsfs}
\usepackage{mdframed}
\usepackage{natbib}
\usepackage[english]{babel}
\usepackage{cases}

\mdfsetup{middlelinecolor=blue,middlelinewidth=2pt,linewidth=0pt,backgroundcolor=blue!15,,roundcorner=10pt}

\allowdisplaybreaks

\setcitestyle{numbers,open={[},close={]}}

\definecolor{red}{rgb}{0.7,0.15,0.15}
\definecolor{green}{rgb}{0,0.5,0}
\definecolor{blue}{rgb}{0,0,0.7}

\usepackage{hyperref}
\hypersetup{colorlinks, linkcolor={red}, citecolor={green}, urlcolor={blue}}
\numberwithin{equation}{section}
\usepackage[english]{cleveref}

\newtheorem{theorem}{Theorem}[section]
\newtheorem{assumption}[theorem]{Assumption}

\newtheorem{example}[theorem]{Example}

\newtheorem{lemma}[theorem]{Lemma}

\newtheorem{definition}[theorem]{Definition}
\newtheorem{remark}[theorem]{Remark}

\newcommand\cB{\mathcal B}

\newcommand\cE{\mathcal E}
\newcommand\cF{\mathcal F}
\newcommand\cG{\mathcal G}
\newcommand\cH{\mathcal H}

\newcommand\cL{\mathcal L}

\newcommand\cS{\mathcal S}
\newcommand\cT{\mathcal T}

\newcommand\cY{\mathcal Y}

\newcommand\sK{\mathscr K}

\newcommand\sM{\mathscr M}
\newcommand\sN{\mathscr N}

\newcommand\sY{\mathscr Y}

\newcommand{\smallertext}[1]{\text{\fontsize{5}{5}\selectfont$#1$}}
\newcommand{\smalltext}[1]{\text{\fontsize{4}{4}\selectfont$#1$}}

\newcommand\fP{\mathfrak{P}}

\def \D{\mathbb{D}}
\def \E{\mathbb{E}}
\def \F{\mathbb{F}}
\def \G{\mathbb{G}}

\def \L{\mathbb{L}}

\def \N{\mathbb{N}}
\def \P{\mathbb{P}}
\def \Q{\mathbb{Q}}
\def \R{\mathbb{R}}
\def \S{\mathbb{S}}

\def \Y{\mathbb{Y}}

\newcommand{\bcdot}{\boldsymbol{\cdot}}

\newcommand{\1}{\mathbf{1}}

\def\d{\mathrm{d}}

\DeclareMathOperator*{\esssup}{ess\,sup}

\begin{document}
	
	\title{Robust Hedging of American Options via Aggregated Snell Envelopes}

	\author{Marco {\sc Rodrigues}\footnote{ETH Zurich, Department of Mathematics, Switzerland, marco.rodrigues@math.ethz.ch. The author thanks Dylan Possama\"i for feedback and insightful discussions.}}
	
	\date{June 17, 2025}
	
	\maketitle
	
	\begin{abstract}
		We construct an aggregator for a family of Snell envelopes in a nondominated framework. We apply this construction to establish a robust hedging duality, along with the existence of a minimal hedging strategy, in a general semi-martingale setting for American-style options. Our results encompass continuous processes, or processes with jumps and non-vanishing diffusion. A key application is to financial market models, where uncertainty is quantified through the semi-martingale characteristics.
		
		\medskip
		\noindent{\bf Key words:}  \vspace{5mm}   American option, characteristics, hedging, model uncertainty, semi-martingales
	\end{abstract}
	

	\section{Introduction}
	
	We study the problem of superhedging American-style options in financial markets where the underlying model is uncertain and potentially incomplete. Compared to classical settings where a single probability measure describes the market, we suppose, in the sense of Knightian uncertainty, that the true model lies within a class $\fP$. The goal is to determine the superreplication price, that is, the smallest initial capital $y \in \R$ that allows for a (self-financing) trading strategy $Z$ such that $y + Z \bcdot S \geq \xi$, $\P$--a.s., for all $\P \in \fP$, where $S$ is the discounted price process of $d$ risky assets, $Z \bcdot S$ is the cumulative gains process, and $\xi$ is the payoff of an option that can be exercised at any moment over the finite time horizon $[0,T]$. We introduce a dual problem, prove the absence of a duality gap, and show that the minimal initial capital can be attained.
	
	\medskip
	The problem of hedging European-style options in incomplete markets has been studied extensively. A foundational result is the optional decomposition introduced by \citeauthor*{el1995dynamic} \cite{el1995dynamic} for continuous It\^o--diffusion-type processes, where perfect hedging is not always possible. This decomposition was extended by \citeauthor*{kramkov1996optional} \cite{kramkov1996optional} to locally bounded semi-martingale models of financial markets, and, subsequently, by \citeauthor*{follmer1997optional} \cite{follmer1997optional} to semi-martingale models with unbounded jumps. \citeauthor*{follmer1997optional2} \cite{follmer1997optional2} then applied techniques used in proving the optional decomposition in the study of hedging problems with constraints. In the context of mathematical finance, \citeauthor*{delbaen1999compactness} \cite{delbaen1999compactness} provided another proof of the optional decomposition theorem in the framework of $\sigma$-martingale measures, which are naturally related to no-arbitrage-type pricing arguments in financial market models with unbounded jumps; see \citeauthor*{delbaen1998fundamental} \cite{delbaen1998fundamental}.
	
	\medskip
	In the works mentioned previously, the probability law that governs the stock price process $S$ is assumed to be exactly known. However, in many situations, the law might only be partially known, and the true probability law lies within a set $\fP$ of possible models. This framework is known as Knightian uncertainty; see \citeauthor*{neufeld2015knightian} \cite{neufeld2015knightian}. In this context, the pricing of European-style contingent claims in continuous-time and for continuous processes has been studied, for example, by \citeauthor*{neufeld2013superreplication} \cite{neufeld2013superreplication}, \citeauthor*{nutz2012superhedging} \cite{nutz2012superhedging}, \citeauthor*{peng2019nonlinear} \cite{peng2019nonlinear}, \citeauthor*{possamai2013robust} \cite{possamai2013robust}, \citeauthor*{soner2011martingale} \cite{soner2011martingale} and \citeauthor*{soner2013dual} \cite{soner2013dual}. The approaches in these works rely on the Doob--Meyer decomposition to construct minimal superhedging strategies; thus, they are restricted to market models that are complete. There is also work by \citeauthor*{avellaneda1995pricing} \cite{avellaneda1995pricing}, \citeauthor*{denis2006theoretical} \cite{denis2006theoretical}, \citeauthor*{dolinsky2014martingale} \cite{dolinsky2014martingale}, \citeauthor*{lyons1995uncertain} \cite{lyons1995uncertain} and \citeauthor*{song2011some} \cite{song2011some}, in which they prove duality formulas for continuous processes in an uncertain volatility setting.
	
	\medskip
	An important advancement was provided by \citeauthor*{nutz2015robust} \cite{nutz2015robust} (see also \citeauthor*{bouchard2015arbitrage} \cite{bouchard2015arbitrage} for a related result in discrete-time), who proved a robust hedging duality for European-style options and showed the existence of a corresponding minimal superhedging strategy under model uncertainty, including settings with jumps. These remarkable results were the main inspiration for this work. In \cite{nutz2015robust}, the idea is to formulate the dynamic superhedging price as the conditional sublinear expectation of the European-style option, whose construction relies on results by \citeauthor*{nutz2013constructing} \cite{nutz2013constructing}; this corresponds to the existence of an aggregator for a family of conditional expectations in a nondominated setting. The superhedging price process is then a $\fP$--super-martingale, and an application of the robust optional decomposition theorem \cite[Theorem~2.4]{nutz2015robust} to the regularisation of this super-martingale yields the minimal superhedging strategy and establishes the desired hedging duality; see also \citeauthor*{biagini2021reduced} \cite{biagini2021reduced} and \citeauthor*{biagini2019reduced} \cite{biagini2019reduced} for other applications of the robust optional decomposition result. 
	
	\medskip
	All the previously mentioned works consider the problem of hedging European-style options. In the discrete-time setting and under model uncertainty, American-style options have been studied by \citeauthor*{aksamit2019robust} \cite{aksamit2019robust}, \citeauthor*{bayraktar2017arbitrage} \cite{bayraktar2017arbitrage,bayraktar2017super} and \citeauthor*{bayraktar2015hedging} \cite{bayraktar2015hedging}. Duality results concerning American-style options in continuous-time under model uncertainty are less studied in the literature. Nonetheless, there is a related duality result by \citeauthor*{nutz2015optimal} \cite{nutz2015optimal} for the subhedging (buyer's) price of an American-style option in continuous-time with continuous processes under volatility uncertainty. Although related, the subhedging problem involves an inf--sup formulation, which requires a different approach; see also \citeauthor*{bayraktar2014robust} \cite{bayraktar2014robust}. However, there are works on optimal stopping under model uncertainty in dominated settings, that is, when there exists a reference measure with respect to which all models are absolutely continuous, by \citeauthor*{arharas2023optimal} \cite{arharas2023optimal}, \citeauthor*{belomestny2016optimal} \cite{belomestny2016optimal} and \citeauthor*{zamfirescu2003optimal}~\cite{zamfirescu2003optimal}. In the setting of nonlinear expectations, optimal stopping problems were studied by \citeauthor*{bayraktar2015optimal}~\cite{bayraktar2015optimal} and \citeauthor*{ekren2014optimal} \cite{ekren2014optimal} in a fully nonlinear setting, \citeauthor*{grigorova2020optimal} \cite{grigorova2020optimal} and \citeauthor*{li2023optimal} \cite{li2023optimal} in a $g$-expectation setting, \citeauthor*{li2025optimal} \cite{li2025optimal} in a $G$-expectation setting, and \citeauthor*{grigorova2024optimal} \cite{grigorova2024optimal,grigorova2025non,grigorova2025nonb} in a setting involving risk measures.
	
	\medskip
	We adopt the methodology of \cite{nutz2015robust} in the context of American-style options, which naturally leads to the problem of aggregating a family of Snell envelopes in a nondominated setting. In the case of European-style options, the dynamic superhedging price corresponds to a suitable regularisation of the conditional sublinear expectations associated with $\fP$ and the option payoff. These conditional sublinear expectations can be constructed as the suprema of expectations of the option over a collection of probability measures, using analytic sets and measurable selection arguments; see \cite{nutz2013constructing}. Since the expectation can be viewed as a mapping from probability measures to integrals, it is measurable with respect to the Borel $\sigma$-algebra generated by the topology of convergence in distribution. Consequently, the resulting supremum inherits an appropriate form of measurability, more precisely, upper semi-analyticity. In contrast, when taking the expectation of a Snell envelope, we are dealing with the expectation of an object that itself depends on the probability measure and whose measurability in the probability parameter is not immediately evident and does not follow from existing results in the literature. We are not integrating a fixed random variable, but rather working with a collection of random variables indexed by probability measures, for which measurability with respect to the probability parameter is not immediately evident. This added dependence constitutes the main technical challenge in constructing the aggregator. However, \cite{nutz2013constructing} serves as a cornerstone for the present work.
	
	\medskip
	We also point out a natural connection to Snell envelopes under sublinear expectations in a nondominated setting, which further highlights the increased complexity compared to the classical Snell envelope under linear expectation. The construction of the aggregator is carried out on the Skorokhod space and relies on the assumption that the American-style claim $\xi$ is of class $(D)$, without requiring any $\omega$-regularity. Although this is closely connected to reflected second-order backward stochastic differential equations (see \citeauthor*{matoussi2013second} \cite{matoussi2013second,matoussi2021corrigendum} and \citeauthor*{matoussi2014second} \cite{matoussi2014second}) in the particular case where the generator is identically zero, we instead work here in full generality under only minimal assumptions on the set $\fP$.
	
	\medskip
	The hedging duality presented here relies on a saturation property of $\fP$ and a technical condition introduced in \Cref{def_dominating_diffusion}.  A closely related and elegant condition, the \textit{dominating diffusion property}, was introduced in \cite[Definition 2.2]{nutz2015robust}. It is well suited to the one-dimensional setting to derive the robust optional decomposition \cite[Theorem 2.4]{nutz2015robust}. In the multi-dimensional case, that is, when the stochastic integral is a vector stochastic integral, however, this condition does not always suffice, as demonstrated by \Cref{exp_robust}. Nonetheless, the proof of \cite[Theorem 2.4]{nutz2015robust} remains valid when the dominating diffusion property is replaced by the slightly stronger, though closely related, condition described in \Cref{def_dominating_diffusion}. We repeat the argument from \cite{nutz2015robust} and clarify where the stronger condition comes into play in the multi-dimensional setting. There is also related work on European-style options, although different in methodology, in the jump setting by \citeauthor*{bouchard2021quasi} \cite{bouchard2021quasi}. There, the authors establish a robust hedging duality in a jump setting without the dominating diffusion assumption, but restricted to more ‘regular’ options, using functional It\^o calculus.
	
	\medskip
	The main difference between our approach and previous results is that we establish a hedging duality for American-style options in a general semi-martingale setting under uncertainty, encompassing both continuous processes and those with jumps. To the best of our knowledge, this duality is new even in the continuous case. We further prove the existence of the minimal initial capital and construct a corresponding superhedging strategy for American-style options.
	
	\medskip 
	The remainder of this paper is organized as follows. In \Cref{sec_setup}, we describe the setup, introduce the notation, and state the assumptions on the family $\fP$ needed to construct the aggregator of the Snell envelopes in \Cref{thm_construction_capital_process}. In \Cref{sec_duality}, we use the existence of the aggregator along with the optional decomposition theorem to derive a robust hedging duality for American-style options in \Cref{thm_duality}. As part of this, we restate and reprove the robust optional decomposition theorem \cite[Theorem 2.4]{nutz2015robust} using our modified dominating diffusion property from \Cref{def_dominating_diffusion}, and highlight the point in the proof where this stronger notion becomes essential in the multi-dimensional setting. Finally, \Cref{sec_proof_aggregation} is devoted to the proof of \Cref{thm_construction_capital_process}.

	\section{Setup and the aggregation result}\label{sec_setup}
	
	\subsection{Preliminaries and assumptions}
	
	Let $d$ be a positive integer. We denote by $\Omega$ the space consisting of all c\`adl\`ag paths $\omega : [0,\infty) \longrightarrow \R^d$ with $\omega_0 = 0$, by $\F = (\cF_t)_{t \in [0,\infty)}$ the canonical (raw) filtration generated by the canonical process $X = (X_t)_{t \in [0,\infty)}$, and we let $\cF \coloneqq \sigma(\cup_{t\in[0,\infty)}\cF_t)$ and $\cF_{0-}\coloneqq \{\varnothing,\Omega\}$. We endow $\Omega$ with the Skorokhod topology, which in turn makes $\Omega$ Polish and yields $\cB(\Omega) = \cF$ (see \citeauthor*{jacod2003limit} \cite[Theorem VI.1.14]{jacod2003limit}). We denote by $\fP(\Omega)$ the collection of all probability measures on $(\Omega,\cF)$ and endow it with the topology of convergence in distribution, that is, the coarsest topology for which $\P \longmapsto \E^\P[f]$ is continuous for every bounded and continuous function $f : \Omega \longrightarrow \R$. Since $\Omega$ is Polish, so is $\fP(\Omega)$ (see \citeauthor*{aliprantis2006infinite} \cite[Theorem 15.15]{aliprantis2006infinite} or \citeauthor*{bertsekas1978stochastic} \cite[Proposition 7.23]{bertsekas1978stochastic}), and the Borel $\sigma$-algebra $\cB(\fP(\Omega))$ is generated by the mappings $\fP(\Omega) \ni \P \longmapsto \P[A] \in [0,1]$, $A \in \cB(\Omega)$ (see \cite[Proposition 7.25]{bertsekas1978stochastic}). For background and further details on the Skorokhod space, see also \citeauthor*{dellacherie1978probabilities} \cite[pp.~145\,ff.]{dellacherie1978probabilities}, \citeauthor*{jacod2003limit} \cite[Chapter~VI]{jacod2003limit} and \citeauthor*{karoui2013capacities} \cite[Section 4.1]{karoui2013capacities}. We frequently use Galmarino's test (see \cite[Theorem IV.99--101]{dellacherie1978probabilities}) to establish $\F_-$, $\F$ or $\F_+$-adaptedness, for example. Furthermore, we also rely on results from \citeauthor{weizsaecker1990stochastic} \cite{weizsaecker1990stochastic} and \citeauthor*{dellacherie1978probabilities} \cite{dellacherie1978probabilities,dellacherie1982probabilities} to deal with non-complete filtrations.
	
	\medskip
	For a generic filtration $\G = (\cG_t)_{t \in [0,\infty)}$ on $\Omega$, with convention $\cG_{0-}\coloneqq \{\varnothing,\Omega\}$, we write $\G_{-} = (\cG_{t-})_{t \in [0,\infty)}$ for the left-limit filtration, that is, $\cG_{t-} \coloneqq \sigma(\cup_{s < t}\cG_s)$ for $t \in(0,\infty)$. We write $\G_+ = (\cG_{t+})_{t\in[0,\infty)}$ for the corresponding right-continuous modification, that is, $\cG_{t+} \coloneqq \cap_{s>t}\cG_s$ for $t \in [0,\infty)$; we still use the convention $\cG_{(0+)-} = \cG_{0-} = \{\varnothing,\Omega\}$. The $\G$-predictable $\sigma$-algebra on $\Omega\times[0,\infty)$ is generated by all real-valued, $\G_-$-adapted processes with left-continuous paths on $(0,\infty)$. Note that the predictable $\sigma$-algebras relative to $\G_{-}$, $\G$ and $\G_{+}$ are all identical. For $\P \in \fP(\Omega)$, we define $\G^\P$ as the filtration obtained by augmenting each $\sigma$-algebra in the original filtration (including $\cG_{0-}$) with all $(\cF,\P)$--null sets. A property holds $\P$--a.s. if it holds outside an $(\cF,\P)$--null set.
	
	\medskip
	The concatenation $\omega\otimes_t\tilde\omega \in \Omega$ of two paths $(\omega,\tilde\omega) \in \Omega \times \Omega$ at $t \in [0,\infty)$ is defined as
	\begin{equation*}
		(\omega \otimes_t \tilde\omega)_s \coloneqq \omega_s \1_{\{0 \leq s < t\}} + (\omega_{t} + \tilde\omega_{s-t})\1_{\{t \leq s\}}, \; s \in [0,\infty).
	\end{equation*}
	
	For $\P \in \fP(\Omega)$ and $t \in [0,\infty)$, there exists a $\P$--a.s. unique stochastic kernel $(\P^t_\omega)_{\omega \in \Omega}$ on $(\Omega,\cF)$ given $(\Omega,\cF_t)$ such that $\E^{\P^\smalltext{t}_\smalltext{\cdot}}[\xi]$ is a version of the conditional expectation $\E^\P[\xi|\cF_t]$ for any (Borel-)measurable function $\xi$ with values in $[-\infty,\infty]$; we use the convention $\infty-\infty = -\infty$ throughout. For the existence of this kernel, we refer to \citeauthor*{kallenberg2021foundations} \cite[Theorem 8.5]{kallenberg2021foundations} or \citeauthor*{dellacherie1978probabilities} \cite[pp.~78\,ff.]{dellacherie1978probabilities}.
	Exactly as in \citeauthor*{stroock1997multidimensional} \cite[page 34]{stroock1997multidimensional}, we choose $(\P^t_\omega)_{\omega \in \Omega}$ in such a way that, additionally,
	\begin{equation*}
		\P^t_\omega\big[ X_{\cdot \land t} = \omega_{\cdot\land t}\big] = 1, \; \omega \in \Omega,
	\end{equation*}
	holds. We then refer to $(\P^t_\omega)_{\omega \in \Omega}$ as the regular conditional probability distribution (r.c.p.d.) of $\P$ given $\cF_t$. We define $\P^{t,\omega} \in \fP(\Omega)$ by
	\[
	\P^{t,\omega}[A] \coloneqq \P^{t}_{\omega}[\omega\otimes_t A], \; A \in \cF, \;  
	\omega\otimes_t A \coloneqq \{\omega\otimes_t\tilde\omega \,|\, \tilde\omega \in A\},
	\]
	and $\xi^{t,\omega} : \Omega \longrightarrow [-\infty,\infty]$ is given by
	\[
	\xi^{t,\omega}(\tilde\omega) \coloneqq \xi(\omega\otimes_t\tilde\omega), \; \tilde\omega \in \Omega.
	\]
	Then 
	\[
	\E^{\P^{t,\omega}}[\xi^{t,\omega}] = \E^{\P^\smalltext{t}_\smalltext{\omega}}[\xi], \; \omega \in \Omega,
	\]
	and thus $\E^{\P^{t,\cdot}}[\xi^{t,\cdot}]$ is also a version of $\E^\P[\xi|\cF_t]$. 
	
	\begin{remark}
		For $t \in [0,\infty)$, it can be readily verified that
		\[
		\P^{t,\omega}[A] = \P\big[X_{t+\cdot}-X_t \in A\big|\cF_t\\
		\big](\omega), \; \textnormal{$\P$--a.e. $\omega \in \Omega$.}
		\]
		This implies that $(\mathbb{P}^{t, \omega})_{\omega \in \Omega}$ is the conditional distribution of the shifted process $X_{t+\cdot} - X_t$ given $\mathcal{F}_t$. 
	\end{remark}
	
	For a process $Z = (Z_s)_{s \in [0,\infty)}$ and $(\omega,t)\in\Omega\times[0,\infty)$, we write $Z^{t,\omega}_{t\smallertext{+}\smallertext{\cdot}} \coloneqq (Z_{t+\cdot})^{t,\omega}$ for the process
	\[
	(Z^{t,\omega}_{t+\cdot})_s(\tilde\omega) = Z_{t+s}(\omega\otimes_t\tilde\omega), \; (\tilde\omega,s) \in \Omega \times [0,\infty).
	\]
	
	\medskip
	We fix a family $(\fP(t,\omega))_{(\omega,t)\in\Omega\times[0,\infty)}$ consisting of subsets $\fP(t,\omega) \subseteq \fP(\Omega)$ that satisfy
	\[
	\fP(t,\omega_{\cdot\land t}) = \fP(t,\omega).
	\]
	Since $\fP(0,\omega)$ is independent of $\omega$, we simply denote it by $\fP$. This set will conceptually serve to formalize the uncertainty in the probability law of the stock price process when establishing the hedging duality.
	
	\medskip
	Our main aggregation result will be proved under the following conditions on this family. They enable us to define the aggregator of the family of Snell envelopes. We recall that a subset of a Polish space is analytic if it is the image of a continuous map defined on a Polish space; see \cite[Section 8.2]{cohn2013measure} or \cite[Section 7.6]{bertsekas1978stochastic} for further background.
	
	\begin{assumption}\label{ass::measures}
		For all $0 \leq t < \infty$ and $\P \in \fP$,
		\begin{enumerate}
			\item[$(i)$] $\{(\omega,\P^\prime) \,|\, \omega \in \Omega,  \; \P^\prime \in \fP(t,\omega)\} \subseteq \Omega \times \fP(\Omega)$ is analytic,
			\item[$(ii)$] $\P^{t,\omega} \in \fP(t,\omega)$, for {\rm$\P$--a.e.} $\omega\in\Omega$,
			\item[$(iii)$] if $\Q$ is a stochastic kernel on $(\Omega,\cF)$ given $(\Omega,\cF_{t})$ such that $\Q(\omega) \in \fP(t,\omega)$, for {\rm$\P$--a.e.} $\omega \in \Omega$, then the probability measure
			\begin{align*}
				\overline{\P}[A] \coloneqq & \int_{\Omega}\int_{\Omega} \big(\1_A\big)^{t,\omega}(\omega^\prime)\Q(\omega;\d\omega^\prime)\P(\d\omega), \; A \in \cF,
			\end{align*}
			belongs to $\fP$.
		\end{enumerate}
	\end{assumption}
	
	\subsection{Main aggregation result}
	
	We now specify a filtration $\G = (\cG_t)_{t \in [0,\infty)}$, which will be used throughout the remainder of this work. For $t \in [0,\infty)$, let
	\[
	\cG_t \coloneqq \cF^\ast_t \lor \sN^\fP,
	\]
	where $\sN^\fP$ denotes the collection of all $A \subseteq \Omega$ that are $(\cF,\P)$--null sets for every $\P\in\fP$, and $\cF^\ast_t$ denotes the universal completion of $\cF_t$; this means that $\cF^\ast_t = \cap_{\P^\smalltext{\prime}}\cF_t(\P^\prime)$, where $\P^\prime$ ranges over all probability measures on $(\Omega,\cF_t)$ and $\cF_t(\P^\prime)$ denotes the $\P^\prime$-completion of $\cF_t$. Note that then $\G^\P_\smallertext{+} = \F^\P_+$, for any $\P\in\fP(\Omega)$, satisfies the usual conditions. 
	
	\medskip
	The first main result of this work is the existence of an aggregator for Snell envelopes in a nondominated setting. We write $\cT_{s,t}(\G_+)$ for the collection of all $\G_+$--stopping times $\tau$ satisfying $0 \leq s \leq \tau \leq t \leq \infty$.
	
	\begin{theorem}\label{thm_construction_capital_process}
		Suppose that $\fP$ is non-empty and that $(\fP(t,\omega))_{(\omega,t)\in\Omega\times[0,\infty)}$ satisfies \textnormal{\Cref{ass::measures}}. Let Let $T \in [0,\infty)$, and let $\xi = (\xi_t)_{t \in [0,T]}$ be a c\`adl\`ag, $\F_+$-adapted process satisfying
		\[
		\sup_{\P\in\fP} \sup_{\tau\in\cT_{\smalltext{0}\smalltext{,}\smalltext{T}}(\G_\smalltext{+})} \E^\P\big[|\xi_\tau|\big] < \infty,
		\]
		and such that for every $(t,\omega) \in [0,T]\times\Omega$, and then for every $\P\in\fP(t,\omega)$, the collection $\{\xi^{t,\omega}_{t+\tau} \,|\, \tau \in \cT_{0,T-t}(\G_+)\}$ is $\P$--uniformly integrable.
		Then there exists a, up to $\fP$--indistinguishability, unique real-valued, c\`adl\`ag, $\G_\smallertext{+}$-adapted process $Y = (Y_t)_{t \in [0,T]}$ satisfying
		\begin{equation}\label{eq_aggregator}
			Y_t = \underset{\bar{\P} \in \fP(\cG_{t\smalltext{+}},\P)}{{\esssup}^\P} \underset{\tau\in\cT_{\smalltext{t}\smalltext{,}\smalltext{T}}(\G_\smalltext{+})}{{\esssup}^{\bar\P}} \E^{\bar{\P}}[\xi_\tau|\cG_{t+}],
			\; \textnormal{$\P$--a.s.}, \; t \in [0,T], \; \textnormal{$\P \in \fP$,}
		\end{equation}
		where $\fP(\cG_{t+},\P) = \big\{\overline{\P}\in\fP \,\big|\, \textnormal{$\overline{\P} = \P$ on $\cG_{t+}$}\big\}$.	In particular, $Y$ is a $(\G^\P_+,\P)$--super-martingale for every $\P\in\fP$ satisfying
		\begin{equation}\label{eq_Y_0_equaly_xi}
			\sup_{\P\in\fP}\E^\P[Y_0] = \sup_{\P\in\fP} \sup_{\tau\in\cT_{\smalltext{0}\smalltext{,}\smalltext{T}}(\G_\smalltext{+})} \E^\P[\xi_\tau].
		\end{equation}
	\end{theorem}
	
	We postpone the proof of this aggregation result to \Cref{sec_proof_aggregation} and proceed directly to its application to the pricing and hedging of American-style options in the following section, after the following important comments.
	
	\begin{remark}
		$(i)$ The process $Y$ constructed in \textnormal{\Cref{thm_construction_capital_process}} can be viewed as the Snell envelope of $\xi$ relative to a sublinear expectation, meaning that it is the smallest nonlinear super-martingale that dominates $\xi$. Since $\overline{\P} = \P$ on $\cG_{t+}$ in \eqref{eq_aggregator}, we find
		\[
		Y_t = \underset{\bar{\P} \in \fP(\cG_{t\smalltext{+}},\P)}{{\esssup}^\P} \underset{\tau\in\cT_{\smalltext{t}\smalltext{,}\smalltext{T}}(\G_\smalltext{+})}{{\esssup}^{\P}} \E^{\bar{\P}}[\xi_\tau|\cG_{t+}] 
		= \underset{\tau\in\cT_{\smalltext{t}\smalltext{,}\smalltext{T}}(\G_\smalltext{+})}{{\esssup}^{\P}}  \underset{\bar{\P} \in \fP(\cG_{t\smalltext{+}},\P)}{{\esssup}^\P} \E^{\bar{\P}}[\xi_\tau|\cG_{t+}], \; \textnormal{$\P$--a.s.}, \; t \in [0,T], \; \P\in\fP.
		\]
		It follows from \textnormal{\cite[Theorem~IV.59]{dellacherie1978probabilities}} that for every $\P \in \fP$ and $\tau \in \cT_{t,T}(\G_+)$, there exists $\tau^\prime \in \cT_{t,T}(\F_+)$ such that $\tau = \tau^\prime$, \textnormal{$\P$--a.s.}, and hence we obtain the representation
		\[
		Y_t = \underset{\tau\in\cT_{\smalltext{t}\smalltext{,}\smalltext{T}}(\F_\smalltext{+})}{{\esssup}^{\P}}  \cE_{t+}(\xi_\tau), \; \textnormal{$\P$--a.s.},  \; t \in [0,T], \; \P\in\fP.
		\]
		Here, $\cE_{t+}(\xi_\tau)$ denotes the right-limit of the conditional sublinear expectations $(\cE_s(\xi_\tau))_{s \in [0,T]}$ of $\xi_\tau$ relative to $\fP$, which satisfy
		\[
		\cE_s(\xi_\tau) = \underset{\bar{\P} \in \fP(\cF_{s},\P)}{{\esssup}^\P} \E^{\bar{\P}}[\xi_\tau|\cF_s], \; \textnormal{$\P$--a.s.},  \; s \in [0,T], \; \P\in\fP.
		\]
		We refer to \textnormal{\cite{neufeld2013superreplication,nutz2013constructing}} and \textnormal{\cite[Section~4.1]{nutz2012superhedging}} for background. The function $Y_\tau$, for $\tau \in \cT_{0,T}(\F_+)$, is not necessarily upper semi-analytic and thus might lie outside the domain of $\cE_t(\,\cdot\,)$, and hence of $\cE_{t+}(\,\cdot\,)$. However, as shown in the proof of \textnormal{\Cref{thm_construction_capital_process}}, we have $Y = Y^0$ outside a $\fP$--polar set, where $Y^0$ is defined as the right-hand limit superior (along a dense subset) of a collection of upper semi-analytic functions $(\cY_s)_{s \in [0,T]}$ that form an $(\F^\ast,\fP)$--super-martingale. It is then straightforward to verify that $Y^0_\tau$ is upper semi-analytic, using the Borel-measurability of $\tau$ and arguments similar to those in the proof of \textnormal{\cite[Lemma~7.30]{bertsekas1978stochastic}}. We may therefore define $\cE_t(Y_\tau) \coloneqq \cE_t(Y^0_\tau)$$;$ in this way, the conditional sublinear expectations of $Y_\tau$ are then uniquely defined outside $\fP$--polar sets. It follows that
		\begin{equation}\label{eq_E_supermartingale}
			\cE_{t+}(Y_\tau) \leq Y_t, \; \textnormal{$\fP$--q.s.}, \; \tau \in \cT_{t,T}(\F_+), \; t \in [0,T].
		\end{equation}
		By the monotonicity of conditional sublinear expectations, we have that for any other c\`adl\`ag, $\G_+$--adapted process $\widehat{Y}$ satisfying \textnormal{\eqref{eq_E_supermartingale}} and $\widehat{Y} \geq \xi$, \textnormal{$\fP$--q.s.}, and for which each $\widehat{Y}_\tau$ admits an upper semi-analytic $\fP$--modification, it must hold that
		\[
		\cE_{t+}(\xi_\tau) \leq \cE_{t+}(\widehat{Y}_\tau) \leq \widehat{Y}_t, \; \textnormal{$\fP$--q.s.}, \; \tau \in \cT_{t,T}(\F_+), \; t \in [0,T].
		\]
		Therefore, $Y \leq \widehat{Y}$ outside a $\fP$--polar set. Hence, $Y$ can be interpreted as the minimal c\`adl\`ag $\cE_+$--super-martingale that dominates $\xi$.
		
		\medskip
		$(ii)$ It is natural to ask whether the above aggregation result extends to irregular processes $\xi$ that are not c\`adl\`ag or are defined on a random and potentially unbounded time horizon, and whether this leads to a corresponding hedging duality result. We discuss this in \textnormal{\Cref{sec_conclusions}}.
	\end{remark}
	
	\section{Hedging duality and optional decomposition}\label{sec_duality}
	
	In this section, we apply the aggregation result for Snell envelopes from \Cref{thm_construction_capital_process} to derive a robust hedging duality in the context of American-style options, that is, options exercisable at any time. We work within a general semi-martingale framework of mathematical finance and refer to the monograph by \citeauthor*{delbaen2006mathematics}~\cite{delbaen2006mathematics} for background, as well as the original articles \cite{delbaen1994general,delbaen1998fundamental}.
	
	\medskip
	Let $S = (S_t)_{t \in [0,\infty)}$ be an $\R^d$-valued, c\`adl\`ag, $\G_+$-adapted process such that $S$ is a $(\G^\P_+,\P)$--semi-martingale for every $\P\in\fP$. Here, $S$ should be understood as playing the role of the discounted price process of $d$ risky assets in a financial market that also includes a bank account whose discounted price process is constantly equal to $1$. We denote by $\L(S,\fP)$ the collection of $\R^d$-valued, $\G$-predictable processes $Z = (Z_t)_{t \in [0,\infty)}$ that are $S$-integrable relative to every $\P\in\fP$ in the sense of vector stochastic integration; see \cite[Section III.6]{jacod2003limit}. For clarity, we regard both the integrand $Z$ and the integrator $S$ as column vectors. We denote the corresponding stochastic integral simply by $Z\bcdot S$, even though the integral implicitly depends on $\P$, except in the case of equivalent measures. A probability measure $\P\in\fP(\Omega)$ is a $\sigma$-martingale measure for $S$, if $S$ is a $(\G^\P_+,\P)$--semi-martingale and for each component $S^i$, there exists a non-decreasing sequence of $\G^\P_+$-predictable subsets $(\Sigma_n)_{n \in \N} \subseteq \Omega\times[0,\infty)$ whose union is $\Omega\times[0,\infty)$ such that $\1_{\Sigma_\smalltext{n}}$ is $S$-integrable and $\1_{\Sigma_\smalltext{n}}\bcdot S^i$ is a $(\G^\P_+,\P)$--uniformly integrable martingale. We refer to \cite[Theorem III.6.41]{jacod2003limit} for other characterisations of $\sigma$-martingales. In the following, we denote the $(\G^\P_+,\P)$--semi-martingale characteristics by $(\mathsf{B}^\P,\mathsf{C}^\P,\nu^\P)$; see \cite[Section II.2]{jacod2003limit}, but also \citeauthor{neufeld2014measurability} \cite{neufeld2014measurability}. 
	
	\medskip
	Instead of working with $\G^\P_+$ above, we could have also simply described all the notions relative to $\G_+$, since $S$ is a semi-martingale with respect to $(\G_+,\P)$ if and only if it is one with respect to $(\G^\P_+,\P)$; see \cite[Proposition~2.2]{neufeld2014measurability}. Moreover, the characteristics also coincide outside a $\P$--null set. Similarly, $S$ is a $\sigma$-martingale relative to $(\G_+,\P)$ if and only if it is one relative to $(\G^\P_+,\P)$; for example, one can apply what is discussed right after \cite[Theorem~IV.78]{dellacherie1978probabilities}. Therefore, switching between $\G^\P_+$ and $\G_+$ in this context is harmless.
	
	\subsection{Dominating diffusion and saturation}
	
	\medskip
	In the following, $\mu \ll \nu$ denotes that the measure $\d\mu$ is absolutely continuous with respect to the measure $\d\nu$, and $\mu \approx \nu$ means that $\mu \ll \nu \ll \mu$. Furthermore, $\textnormal{Tr}(\mathsf{C})$ denotes the trace of a square matrix $\mathsf{C}$, and $\S^d_+$ denotes the space of real positive semi-definite $d \times d$ matrices.
	
	\begin{lemma}\label{prop::dominating_diffusion_equivalences}
		Let $\P\in \fP(\Omega)$ be such that $S$ is a $(\G^\P_+,\P)$--semi-martingale, and let $(\mathsf{B}^\P,\mathsf{C}^\P,\nu^\P)$ be the corresponding characteristics. Let $\mathsf{c}^\P$ be the \textnormal{$\P\otimes\mathrm{d}\textnormal{Tr}(\mathsf{C}^\P)$--a.e.} uniquely defined $\S^{d}_\smallertext{+}$-valued process such that $\mathsf{C}^\P = \mathsf{c}^\P \bcdot \textnormal{Tr}(\mathsf{C}^\P)$, \textnormal{$\P$--a.s.}, which is a factorisation of the second characteristic. Then the conditions
		\begin{enumerate}
			\item[$(i)$] $(|x|^2 \land 1)\ast \nu^\P \ll \1_{\{\textnormal{det}(\mathsf{c}^\smalltext{\P})>0\}}\bcdot\textnormal{Tr}(\mathsf{C}^\P)$, \textnormal{$\P$--a.s.},
			\item[$(ii)$] $(|x|^2 \land 1)\ast \nu^\P \ll \1_{\{\textnormal{det}(\mathsf{c}^\smalltext{\P}) \neq 0\}}\bcdot\textnormal{Tr}(\mathsf{C}^\P)$, \textnormal{$\P$--a.s.},
			\item[$(iii)$] $(|x|^2 \land 1)\ast \nu^\P \ll \1_{\{|\textnormal{det}(\mathsf{c}^\smalltext{\P})|>0\}}\bcdot\textnormal{Tr}(\mathsf{C}^\P)$, \textnormal{$\P$--a.s.},
			\item[$(iv)$] $(|x|^2 \land 1)\ast \nu^\P \ll \textnormal{det}(\mathsf{c}^\smalltext{\P})\bcdot\textnormal{Tr}(\mathsf{C}^\P)$, \textnormal{$\P$--a.s.},
			\item[$(v)$] $(|x|^2 \land 1)\ast \nu^\P \ll |\textnormal{det}(\mathsf{c}^\P)|\bcdot\textnormal{Tr}(\mathsf{C}^\P)$, \textnormal{$\P$--a.s.},
		\end{enumerate}
		are equivalent.
	\end{lemma}
	
	\begin{proof}
		Since the determinant of any matrix in $\S^d_+$ is non-negative, the equivalences follow immediately.
	\end{proof}
	
	The following notion turns out to be a useful technical condition in deriving a robust version of Kramkov's optional decomposition result.
	
	\begin{definition}\label{def_dominating_diffusion}
		Let $\P\in\fP(\Omega)$ be such that $S$ is a $(\G^\P_+,\P)$--semi-martingale. If any of the equivalent conditions in \textnormal{\Cref{prop::dominating_diffusion_equivalences}} holds for $\P$, then $S$ is said to have dominating diffusion relative to $\P$.
	\end{definition}
	
	\begin{remark}
		In the one-dimensional case, the notion of dominating diffusion introduced in \textnormal{\Cref{def_dominating_diffusion}} coincides with the original definition given in \textnormal{\cite[Definition 2.2]{nutz2015robust}}, where the term was first introduced. In the multi-dimensional setting, however, our notion is slightly stronger. Upon careful examination of the proof of the robust optional decomposition \textnormal{\cite[Theorem 2.4]{nutz2015robust}}, it appears that a stronger condition than that provided by \textnormal{\cite[Definition 2.2]{nutz2015robust}} is required at a particular step to complete the argument. For completeness, we include a proof that closely follows the structure of the proof of \textnormal{\cite[Theorem 2.4]{nutz2015robust}}, with a minor but essential modification to the relevant step, and makes use of \textnormal{\Cref{def_dominating_diffusion}}. We then present a counterexample to demonstrate the necessity of the stronger condition.
	\end{remark}
	
	Before stating the hedging duality and the decomposition result, we introduce the following notion of saturation from \cite[Section 2]{nutz2015robust}.
	
	\begin{definition}
		The collection $\fP$ is said to be saturated for $S$ if, for any $\P' \in \fP(\Omega)$ that is a $\sigma$-martingale measure for $S$ and satisfies $\P' \approx \P$ for some $\P \in \fP$, we have $\P' \in \fP$.
	\end{definition}
	
	\subsection{Hedging duality for American options}
	
	The next result establishes the main hedging duality and the existence of a minimal hedging strategy for American-style options. We denote by $\L_\textnormal{adm}(S,\fP)$ the collection of integrands $Z \in \L(S,\fP)$ for which $Z \bcdot S$ is a $(\G^\P_+,\P)$--super-martingale for each $\P \in \fP$. The subscript `$\textnormal{adm}$' alludes to the admissibility of integrands in mathematical finance; see, for example, \cite[Chapter~8]{delbaen2006mathematics}.
	
	\begin{theorem}\label{thm_duality}
		Suppose that $(\fP(t,\omega))_{(\omega,t)\in\Omega\times[0,\infty)}$ satisfies \textnormal{\Cref{ass::measures}}, that $\fP$ is a non-empty and saturated set of $\sigma$-martingale measures for $S = (S_t)_{t \in [0,\infty)}$, and that $S$ satisfies the dominating diffusion property relative to every $\P\in\fP$. Let $T \in [0,\infty)$, and let $\xi = (\xi_t)_{t \in [0,T]}$ be a c\`adl\`ag, $\F_+$--adapted process satisfying
		\[
		\sup_{\P\in\fP} \sup_{\tau \in \cT_{\smalltext{0}\smalltext{,}\smalltext{T}}(\G_\smalltext{+})} \E^\P\big[|\xi_\tau|\big] < \infty,
		\]
		and such that for every $(t,\omega) \in [0,T]\times\Omega$, and then for every $\P\in\fP(t,\omega)$, the collection $\{\xi^{t,\omega}_{t+\tau} | \tau \in \cT_{0,T-t}(\G_+)\}$ is $\P$--uniformly integrable.	Then
		\begin{equation}\label{eq_duality}
			\sup_{\P\in\fP} \sup_{\tau \in \cT_{\smalltext{0}\smalltext{,}\smalltext{T}}(\G_\smalltext{+})} \E^\P[\xi_\tau] = \inf\big\{ y \in \R \,\big|\, y + Z \bcdot S \geq \xi, \; \textnormal{$\P$--a.s., for all $\P\in\fP$},\; \textnormal{for some $Z \in \L_\textnormal{adm}(S,\fP)$} \big\}.
		\end{equation}
		Moreover, the infimum on the right-hand side is attained.
	\end{theorem}
	
	We postpone the proof of \Cref{thm_duality} to \Cref{sec_proof_duality}, as it relies on a robust optional decomposition introduced in \Cref{sec_robust_decomposition}.
	
	\begin{remark}
		$(i)$ If $\xi$ were bounded from below, then we could replace the set $\L_\textnormal{adm}(S,\fP)$ with $\L(S,\fP)$ in \textnormal{\eqref{eq_duality}}, as $Z \bcdot S$ would then be bounded from below and hence a $\P$--super-martingale by a consequence of a result by Ansel--Stricker $($see \textnormal{\cite[Corollary 7.3.8]{delbaen2006mathematics}}$)$, since $S$ is a $\P$--$\sigma$-martingale.
		
		\medskip
		$(ii)$ The assumptions on the family $(\fP(t,\omega))_{(\omega,t) \in \Omega \times [0,\infty)}$ ensure that we can construct the minimal capital process $Y$ corresponding to the American-style option $\xi$ with \textnormal{\Cref{thm_construction_capital_process}}. Then, $Y$ will be a $\fP$--super-martingale, and the assumptions on $\fP$ allow us to apply the robust optional decomposition from \textnormal{\Cref{thm_optional_decomposition}}.
	\end{remark}
	
	\subsection{Examples}
	
	We provide a brief overview of two example classes where our hedging duality is relevant: random $G$-expectation and nonlinear L\'evy processes.
	
	\subsubsection{Random \emph{G}-expectation}
	
	In the case of continuous processes, the above result naturally applies to the random $G$-expectation setting introduced in \textnormal{\cite{nutz2013random,nutz2013constructing}}. We outline \textnormal{\cite[Example 2.1]{neufeld2013superreplication}} for clarity and completeness. For each $(\omega,t) \in \Omega \times [0,\infty)$, let $\mathbf{D}_t(\omega) \subseteq \R^{d \times d}$. Suppose that for every $t \in [0,\infty)$,
	\[
	\{(\omega,s,A) \,|\, (\omega,s) \in \Omega \times [0,t], \; A \in \mathbf{D}_s(\omega)\} \in \cF_t \otimes \cB([0,t]) \otimes \cB(\R^{d\times d}).
	\]
	Then, let $\P_0$ be the Wiener measure on $(\Omega,\cF)$, and consider the class $\fP_S \subseteq \fP(\Omega)$ (see also \cite[Section 2]{soner2013dual}) consisting of all probability measures $\P^\alpha$ defined through
	\[
	\P^\alpha \coloneqq \P_0 \circ \bigg(\int_0^\cdot \alpha^{1/2}_s\d X_s\bigg)^{-1},
	\]
	for some $\F$-predictable process $\alpha: \Omega \times [0,\infty)\longrightarrow \S^d_{++}$, where $\S^d_{++}$ denotes the set of real positive definite $d \times d$ matrices, satisfying $\int_0^T\|\alpha^{1/2}_s\|^2 \, \d s < \infty$, $\P_0$--a.s., for all $T\in [0,\infty)$; here, $\|\cdot\|$ denotes the Euclidean norm in $\R^{d\times d}$. We then define $\fP(t,\omega)$ as the collection of all $\P \in \fP_S$ for which the second characteristic $\mathsf{C}^\P$ of $X$ satisfies
	\[
	\frac{\d\mathsf{C}^\P_s}{\d s}(\tilde{\omega}) \in \mathbf{D}_{t+s}(\omega\otimes_t\tilde{\omega}), \; \textnormal{$\P\otimes\d s$--a.e. $(\tilde{\omega},s)\in\Omega\times[0,\infty)$.}
	\]
	It then follows from \textnormal{\cite[Corollary 2.6]{neufeld2013superreplication}} that \textnormal{\Cref{ass::measures}} holds. The dominating diffusion property holds immediately because we are considering measures for which $X$ does not jump outside null sets. Moreover, since the measures in $\fP_S$ satisfy the martingale representation property relative to $X$ (see \textnormal{\cite[Lemma 4.1]{neufeld2013superreplication}}, \textnormal{\cite[Lemma 4.4]{nutz2012superhedging}}, or \textnormal{\cite[Lemma 8.2]{soner2011quasi}}), it follows from Girsanov's theorem \cite[Theorem III.3.24]{jacod2003limit} and from \textnormal{\cite[Theorem III.4.29]{jacod2003limit}} that any $\sigma$-martingale measure $\P\in\fP(\Omega)$ for $X$ with $\P\approx\P^\alpha$ for some $\P^\alpha \in\fP_S$ must satisfy $\P=\P^\alpha$. Thus, $\fP$ is saturated in this case.
	
	\subsubsection{Nonlinear L\'evy processes}
	A natural application of the above result in the case of models with jumps is to nonlinear L\'evy processes, as introduced by \textnormal{\citeauthor{neufeld2016nonlinear} \cite{neufeld2016nonlinear}}, or to nonlinear exponential L\'evy processes$;$ see \textnormal{\cite[Section 4]{nutz2015robust}}. We briefly describe the former class of processes here. Let $S = X$ be the canonical process on $\Omega$. Suppose that each $\fP(t,\omega) \equiv \fP_\Theta$, where $\fP_\Theta$ denotes the collection of all semi-martingale measures $\P$ on $\Omega$ for which the characteristics of $S = X$ are absolutely continuous with respect to Lebesgue measure, and the corresponding differential characteristics 
	\[
	(\mathsf{b}^\P_t, \mathsf{c}^\P_t, \mathsf{F}^\P_t(\d x)) \in \Theta, \; \textnormal{$\P\otimes \d t$--a.s.},
	\]
	where $\Theta \subseteq \R^d \times \S^d_+ \times \cL$ is a Borel-measurable subset. Here, $\cL$ denotes the space of all L\'evy measures\footnote{There exists a natural Borel $\sigma$-algebra on $\cL$$;$ see \cite[Section 2.1]{neufeld2014measurability}.} on $\R^d$, that is, the collection of all measures $F$ on $(\R^d, \cB(\R^d))$ satisfying $\int_{\R^d} (|x|^2 \land 1) F(\d x) < \infty$ and $F(\{0\}) = 0$. Then $\fP_\Theta$ is known to satisfy \textnormal{\Cref{ass::measures}}$;$ see \textnormal{\cite[Theorem 2.1.$(i)$]{neufeld2016nonlinear}}. 
	
	\medskip
	For conditions on $\Theta$ under which $\fP_\Theta$ additionally satisfies the property of being a saturated set of $\sigma$-martingale measures for $S = X$, and under which the dominating diffusion property holds in the sense of \textnormal{\Cref{def_dominating_diffusion}} for each measure in $\fP_\Theta$, we refer to \textnormal{\cite[Lemma 4.2]{nutz2015robust}}. Even though our notion of dominating diffusion in \textnormal{\Cref{def_dominating_diffusion}} is stronger than that of \textnormal{\cite[Definition 2.2]{nutz2015robust}}, the proof of \textnormal{\cite[Lemma 4.2]{nutz2015robust}} indeed requires our definition.
	
	\subsection{A robust optional decomposition}\label{sec_robust_decomposition}
	
	\medskip
	The proof of the hedging duality is based on the following robust optional decomposition, which appears as Theorem 2.4 in \cite{nutz2015robust}, with the only difference in our formulation being the use of the notion of dominating diffusion introduced in \Cref{def_dominating_diffusion}.
	
	\begin{theorem}\label{thm_optional_decomposition}
		Let $\fP$ be a non-empty and saturated set of $\sigma$-martingale measures for $S$ such that $S$ satisfies the dominating diffusion property in the sense of \textnormal{\Cref{def_dominating_diffusion}} for every $\P\in\fP$. Suppose that $Y = (Y_t)_{t \in [0,\infty)}$ is a real-valued, c\`adl\`ag, $\G_+$-adapted, $(\G^\P_+,\P)$--local super-martingale for every $\P\in\fP$. There exists a $\G$-predictable process $Z \in \L(S,\fP)$ such that $Y - Y_0 - Z \bcdot S$ is $\P$--a.s. non-increasing for each $\P\in\fP$. 
	\end{theorem}
	
	\begin{remark}
		$(i)$ Instead of working with $\sigma$-martingale measures, we could also work with local martingale measures for $S$. Rather than referring to the optional decomposition of \textnormal{\cite[Theorem 5.1]{delbaen1999compactness}} in the proof of \textnormal{\Cref{thm_optional_decomposition}}, we would invoke \textnormal{\cite[Theorem 1]{follmer1997optional}}.
		
		\medskip
		$(ii)$ It is not relevant for the statement or the proof that we are working on the canonical space of c\`adl\`ag paths. The only important property we need is that the second characteristic of the joint process $(S,Y)$ can be defined independently of the probability law. As the proof of \textnormal{\cite[Proposition 6.6]{neufeld2014measurability}} shows, this can be done without any assumption on the underlying filtered probability space. Moreover, the collection of probability measures $\fP$ does not need to originate from an auxiliary family $(\fP(t,\omega))_{(\omega,t) \in \Omega \times [0,\infty)}$.
	\end{remark}
	
	\begin{proof}
		We adopt the structure and notation of the proof of \cite[Theorem 2.4]{nutz2015robust}, introducing a minor but essential adjustment at a specific step, where we rely on the stronger notion of dominating diffusion given in \Cref{def_dominating_diffusion}. We write $\mathsf{C}^{(S,Y)}$ for the $\G$-predictable process with values in $\S^{d+1}_+$, which coincides for every $\P\in\fP$ with the second characteristic of the joint process $(S,Y)$. We denote by $\mathsf{C}^S$ the submatrix process of $\mathsf{C}^{(S,Y)}$ that corresponds to the second characteristic of $S$. Moreover, we denote by $\mathsf{C}^{SY}$ the $\R^d$-valued process corresponding to the quadratic co-variation of $S$ and $Y$. Let $\mathsf{A} \coloneqq \textnormal{Tr}(\mathsf{C}^S)$, and then let $\d \mathsf{C}^S = \mathsf{c}^S \d\mathsf{A}$ and $\d \mathsf{C}^{SY} = \mathsf{c}^{SY}\d\mathsf{A}$ outside an $\fP$--polar set, for $\G$-predictable processes $\mathsf{c}^S$ and $\mathsf{c}^{SY}$ with values in $\S^d_+$ and $\R^d$, respectively. We argue that the process
		\[
		Z \coloneqq (\mathsf{c}^S)^\oplus \mathsf{c}^{SY},
		\]
		where $(\mathsf{c}^S)^{\oplus}$ denotes the Moore--penrose pseudo-inverse of $\mathsf{c}^S$, serves the intended purpose. We fix $\P\in\fP$, and apply the classical optional decomposition theorem \cite[Theorem 5.1]{delbaen1999compactness} with respect to $(\G^\P_+,\P)$ and obtain the existence of a $\G^\P_+$-predictable and $S$-integrable process $Z^\P$ such that
		\[
		Y - Y_0 - Z^\P\bcdot S
		\]
		is $\P$--a.s. non-increasing. The continuous local martingale part of $Y$ relative to $\P$ (and $\G^\P_+$) is then given by
		\[
		Y^{c,\P} = Z^\P\bcdot S^{c,\P},
		\]
		which then yields
		\[
		\mathsf{c}^{SY} = \mathsf{c}^S Z^\P, \; \textnormal{$\P\otimes \d\mathsf{A}$--a.e.}
		\]
		It is then straightforward to check that
		\[
		(Z - Z^\P)^\top \mathsf{c}^S (Z-Z^\P) = 0, \; \textnormal{$\P\otimes \d\mathsf{A}$--a.e.},
		\]
		which implies that $Z$ is $S^c$-integrable in the sense of vector stochastic integration of continuous local martingales (see \cite[Section III.4a, p. 179ff.]{jacod2003limit}), and that therefore $Z \bcdot S^{c,\P} = Z^\P\bcdot S^{c,\P}$ outside a $\P$--null set.
		
		\medskip
		It remains to show that $Z$ is $(S-S^{c,\P})$-integrable and that $Z \bcdot (S-S^{c,\P}) = Z^\P\bcdot(S-S^{c,\P})$, $\P$--a.s., to conclude. To do this, the proof of \cite[Theorem 2.4]{nutz2015robust} defines the process $\mathsf{A}^\ast$ through
		\[
		\frac{\d\mathsf{A}^\ast}{\d\mathsf{A}} = \min_{1\leq i \leq d} \frac{\d(\mathsf{C}^S)^{ii}}{\d\mathsf{A}},
		\]
		which then still yields $\mathsf{c}^{SY} = \mathsf{c}^S Z^\P$, $\P\otimes\d\mathsf{A}^\ast$--a.e., since $\mathsf{A}^\ast \ll \mathsf{A}$ outside a $\P$--null set. However, $\mathsf{c}^S$ is not invertible $\P\otimes \d\mathsf{A}^\ast$--a.e. in general, since a positive semi-definite matrix with strictly positive entries on the diagonal may not be invertible; as an example, consider the $2\times 2$ matrix whose entries are all $1$'s. We give a more concrete counterexample after the proof. Instead, let us define
		\[
		\mathsf{A}^\ast \coloneqq \1_{\{\textnormal{det}(\mathsf{c})\neq 0\}} \bcdot \mathsf{A}.
		\]
		Then we still have $\mathsf{c}^{SY} = \mathsf{c}^S Z^\P$, $\P\otimes\d\mathsf{A}^\ast$--a.e., and $\mathsf{c}^S$ is now indeed invertible $\P\otimes\d\mathsf{A}^\ast$--a.e., which yields
		\begin{equation}\label{eq_Z_equals_Z_P}
			Z = (\mathsf{c}^S)^{-1} \mathsf{c}^{SY} = Z^\P, \; \textnormal{$\P\otimes\d\mathsf{A}^\ast$--a.e.}
		\end{equation}
		Our dominating diffusion property from \Cref{def_dominating_diffusion} still implies that $\mathsf{A}^\ast$ dominates the characteristics of $S - S^{c,\P}$. It then follows from \cite[Theorem III.6.30]{jacod2003limit} that $Z$ is indeed $(S - S^{c,\P})$-integrable relative to $\P$, and from \cite[Proposition~IX.5.3]{jacod2003limit} that the $\P$--semi-martingale characteristics of
		\[
		N \coloneqq Z \bcdot (S - S^{c,\P}) - Z^\P \bcdot (S - S^{c,\P}) = (Z - Z^\P) \bcdot (S - S^{c,\P}), \; \textnormal{$\P$--a.s.},
		\]
		are identically zero. Hence, $N = 0$, $\P$--a.s., and, therefore,
		\[
		Z \bcdot (S - S^{c,\P}) = Z^\P \bcdot (S - S^{c,\P}), \; \textnormal{$\P$--a.s.}
		\]
		This concludes the proof.
	\end{proof}
	
	We shed some light on the modified dominating diffusion property from \Cref{def_dominating_diffusion} by considering an example of a $\sigma$-martingale measure for $S$ that satisfies the dominating diffusion property in the sense of \cite[Definition 2.2]{nutz2015robust} (as opposed to the version used in this work), but for which the $\S^d_+$-valued process $\mathsf{c}^S$ in the preceding proof is never invertible.
	
	\begin{example}\label{exp_robust}
		Let $d = 2$, and let $S = X$ be the canonical process on $\Omega$. Consider the probability law $\P\in\fP(\Omega)$ for which $S = (S^1,S^2)$ is a semi-martingale with $\P$-characteristics
		\[
		\mathsf{B}^{S}_t = \begin{pmatrix} 0 \\ 0 \end{pmatrix}, \; 
		\mathsf{C}^S_t = \begin{pmatrix} t & t \\ t & t \end{pmatrix}, \; \nu(\d t, \d x) = \delta_{(1,1)}(\d x)\d t,
		\]
		for the truncation function $h = (h^1,h^2) : \R^2 \longmapsto \R^2$ given by $h(x) = x \1_{\{|x| \leq 1\}}$. The existence of $\P$ can be deduced from the results in \textnormal{\cite[Section IX.2d, p.~535~ff.]{jacod2003limit}}. Note that $S$ is a $\P$--$\sigma$-martingale by \textnormal{\cite[Proposition~III.6.35]{jacod2003limit}}.
		In the notation of the preceding proof, we have
		\[
		\d\mathsf{A} = 2 \, \d t, \; \mathsf{c}^S = \begin{pmatrix} 1/2 & 1/2 \\ 1/2 & 1/2 \end{pmatrix},
		\]
		so that $\mathsf{c}^S$ is never invertible, even though, under $\P$, the process $S$ is a $\P$--$\sigma$-martingale satisfying the dominating diffusion property in the sense of \textnormal{\cite[Definition 2.2]{nutz2015robust}}, but not in the sense of our \textnormal{\Cref{def_dominating_diffusion}}. It therefore appears that, for the proof of \textnormal{\Cref{thm_optional_decomposition}} to go through, the dominating diffusion property in the sense of \textnormal{\Cref{def_dominating_diffusion}} is essential.
	\end{example}
	
	We conclude this section with the proof of the hedging duality.
	
	\subsection{Proof of Theorem \ref{thm_duality}}\label{sec_proof_duality}

	We closely follow the proof of \cite[Theorem 3.2]{nutz2015robust}. We first prove the inequality $(\leq)$. Fix $\P\in\fP$, and let $y \in \R$ be such that there exists $Z \in \L_\textnormal{adm}(S,\fP)$ with
	\[
	\xi_t \leq y + Z \bcdot S_t, \; t \in [0,T], \; \textnormal{$\P$--a.s.}
	\]
	Since $Z \bcdot S$ is a $(\G^\P_+,\P)$--super-martingale starting at zero, we find
	\[
	\E^\P[\xi_\tau] \leq y + \E^\P[Z\bcdot S_\tau] \leq y,
	\]
	for any $\tau \in \cT_{0,T}(\G_+)$. Taking the supremum over $\tau$ and $\P$ on the left-hand side then yields the inequality ($\leq$) in \eqref{eq_duality}.
	
	\medskip
	We turn to the converse inequality $(\geq)$. Let $Y = (Y_t)_{t \in [0,T]}$ be the process constructed in \Cref{thm_construction_capital_process}. According to \Cref{thm_optional_decomposition}, there exists $Z = Z\1_{\llbracket 0,T\rrbracket} \in \L(S,\fP)$ such that $Y - Z\bcdot S$ is $\P$--a.s. non-increasing for every $\P\in\fP$, which implies
	\[
	\xi_t \leq Y_t \leq Y_0 + Z\bcdot S_t, \; t \in [0,T], \; \textnormal{$\P$--a.s.}, \; \P\in\fP.
	\]
	It is enough to show that $Y_0 \leq \sup_{\P\in\fP} \sup_{\tau \in \cT_{0,T}(\G_\smalltext{+})} \E^{\P}[\xi_\tau] \eqqcolon \cY_0$, $\P$--a.s. for every $\P\in\fP$. Indeed, this would imply that $\cY_0 \in \R$ belongs to the set on the right-hand side of \eqref{eq_duality}, thereby establishing the reverse inequality ($\geq$) in \eqref{eq_duality}, and since then
	\[
	\E^\P[\xi_T|\cG_{t+}] - \cY_0 \leq Y_t - \cY_0 \leq Z\bcdot S_t, \; \textnormal{$\P$--a.s.},\;  t \in [0,T], \; \P\in\fP,
	\]
	it follows from \cite[Theorem 14.5.13]{delbaen2006mathematics} that $Z\bcdot S$ is a $\P$--super-martingale for every $\P\in\fP$; thus $Z \in \L_\textnormal{adm}(S,\fP)$.
	
	\medskip
	To this end, we will show that
	\begin{equation}\label{sup_radon_nikodym}
		Y_0 \leq \sup \big\{\E^\Q[Y_0] \,\big|\, \Q \in \fP(\Omega,\cG_{0+}), \, \Q \approx \P|_{\cG_{\smalltext{0}\smalltext{+}}}\big\}  
		\leq \cY_0, \; \textnormal{$\P$--a.s.}, \; \P\in\fP,
	\end{equation}
	where $\fP(\Omega,\cG_{0+})$ denotes the collection of all probability measures on $(\Omega,\cG_{0+})$ and $\P|_{\cG_{\smalltext{0}\smalltext{+}}}$ the restriction of $\P$ to $\cG_{0+}$. The second inequality can be argued as follows. For $\Q \in \fP(\Omega,\cG_{0+})$ with $\Q \approx \P|_{\cG_{\smalltext{0}\smalltext{+}}}$, we define $\overline{\P} \in \fP(\Omega)$ by $\d\overline{\P} \coloneqq \frac{\d\Q}{\d\P|_{\cG_{\smalltext{0}\smalltext{+}}}}\d\P$, and note that the $(\G^\P_+,\overline{\P})$--semi-martingale characteristics of $S$ are the same as the $(\G^\P_+,\P)$--semi-martingale characteristics, by Girsanov's theorem (see \cite[Theorem~III.3.24]{jacod2003limit}). Since the property of being a $\sigma$-martingale can be read from the characteristics (see \cite[Proposition~III.6.35.b)]{jacod2003limit}), it follows that $\overline{\P}$ is a $\sigma$-martingale measure for $S$. Thus, by the saturation property and the fact that $\overline{\P} \approx \P$, we have $\overline{\P} \in \fP$. Therefore,
	\[
	\E^\Q[Y_0] = \E^{\bar{\P}}[Y_0] \leq \sup_{\P\in\fP}\E^\P[Y_0] \leq \cY_0.
	\]
	The last inequality follows from \eqref{eq_Y_0_equaly_xi}. 
	
	\medskip
	It remains to argue that
	\[
	Y_0 \leq \sup \big\{\E^\Q[Y_0] \,\big|\, \Q \in \fP(\Omega,\cG_{0+}), \, \Q \approx \P|_{\cG_{\smalltext{0}\smalltext{+}}}\big\}, \; \textnormal{$\P$--a.s.}, \P\in\fP,
	\]
	holds, which we do by contradiction. For ease of notation, we abbreviate the right-hand side to $\sup_{\Q\approx\P}\E^\Q[Y_0]$. Fix $\P\in\fP$, and suppose that
	\[
	\P[A] > 0, \; \textnormal{where} \; A \coloneqq \bigg\{Y_0 > \sup_{\Q\approx\P}\E^\Q[Y_0]\bigg\} \in \cG_{0+}.
	\]
	Then
	\[
	\E^{\P}[\1_A Y_0] >  \P[A] \sup_{\Q\approx\P}\E^\Q[Y_0],
	\]
	and thus
	\[
	\frac{\E^{\P}[\1_A Y_0]}{\P[A]} > \sup_{\Q\approx\P}\E^\Q[Y_0].
	\]
	By letting $\Q_\varepsilon \coloneqq \varepsilon \P|_{\cG_{\smalltext{0}\smalltext{+}}} + (1-\varepsilon)\P|_{\cG_{\smalltext{0}\smalltext{+}}}[\cdot\cap A]/\P[A]$, we find
	\[
	\E^{\Q_\smalltext{\varepsilon}}[Y_0] = \varepsilon\E^\P[Y_0] + (1-\varepsilon)\frac{\E^{\P}[\1_A Y_0]}{\P[A]} > \sup_{\Q\approx\P}\E^\Q[Y_0],
	\]
	for $\varepsilon \in (0,\infty)$ small enough. Since $\Q_\varepsilon$ is a measure on $\cG_{0+}$ and equivalent to $\P|_{\cG_{\smalltext{0}\smalltext{+}}}$, we obtain a contradiction. Therefore, $\P[A] = 0$, which implies $Y_0 \leq \cY_0$, $\P$--a.s., for all $\P \in \fP$. This concludes the proof.
	
	\section{Proof of Theorem \ref{thm_construction_capital_process}}\label{sec_proof_aggregation}
	
	This section is dedicated to proving our main aggregation result, \Cref{thm_construction_capital_process}. The proof is organized in several steps, each of which builds on the previous one.
	
	\medskip
	Let
	\begin{equation*}
		\cY_t(\omega) \coloneqq \sup_{\P\in\fP(t,\omega)} \E^\P\big[Y^{\P}_0(T-t,\xi^{t,\omega}_{t+\cdot})\big], \; t \in [0,T], \; \omega \in \Omega,
	\end{equation*}
	where $Y^{\P}(T-t,\xi^{t,\omega}_{t+\cdot})$ denotes the $(\F_+,\P)$--Snell envelope (see \cite[Appendix I]{dellacherie1982probabilities}) of $\xi^{t,\omega}_{t+\cdot} = (\xi^{t,\omega}_{t+s})_{s \in [0,T-t]}$, that is, $Y^{\P}(T-t,\xi^{t,\omega}_{t+\cdot})$ is the right-continuous $(\F_+,\P)$--super-martingale satisfying
	\[
	Y^{\P}_s(T-t,\xi^{t,\omega}_{t+\cdot})=  \underset{\tau\in\cT_{\smalltext{s}\smalltext{,}\smalltext{T}\smalltext{-}\smalltext{t}}(\F_\smalltext{+})}{{\esssup}^\P} \E^\P[\xi^{t,\omega}_{t+\tau}|\cF_{s+}], \; \textnormal{$\P$--a.s.}, \; s \in [0,T-t].
	\]
	Note that we could replace the filtration $\F_+$ in the above line by $\G_+$ since every stopping time relative to the latter filtration is $\P$--a.s. equal to a stopping time of the former filtration; see \cite[Theorem~IV.59]{dellacherie1978probabilities}.
	
	\medskip
	\textbf{Step 1:} Fix $t\in[0,T]$. We show that $\cY_t$ is upper semi-analytic and $\cF^\ast_t$-measurable. The former means that the set $\{\omega\in\Omega \,|\,\cY_t(\omega) > \lambda\}$ is an analytic subset of $\Omega$ for every $\lambda \in \R$.	To deduce upper semi-analyticity of $\cY_t$, it suffices by \Cref{ass::measures}.$(i)$ and \cite[Proposition 7.47]{bertsekas1978stochastic} to show that, for fixed $t\in[0,T]$, the map
	\[
	D \ni (\omega,\P) \longmapsto \E^\P\big[Y^{\P}_0(T-t,\xi^{t,\omega}_{t+\cdot})\big] \in [-\infty,\infty]
	\]
	is Borel-measurable, where $D = \{(\omega,\P) \,|\, \omega\in\Omega,\, \P\in\fP(t,\omega)\}$; Borel-measurable functions are upper semi-analytic.
	
	\medskip
	Fix $\omega \in \Omega$ and then $\P\in\fP(t,\omega)$. It follows from \cite[Corollary 2.2 and Theorem 2.13]{klimsiak2015reflected} that there exists a, up to $\P$--indistinguishability, unique solution $(\sY,\sM,\sK)$ to the reflected BSDE relative to $(\F_+,\P)$ with terminal time $T-t$, obstacle $\xi^{t,\omega}_{t+\cdot}$, and generator identically zero
	\begin{equation}\label{eq_reflected_bsde}
		\begin{cases}
			\displaystyle \sY_s = \xi^{t,\omega}_{T} - \int_s^{T-t} \d\sM_r + \int_s^{T-t}\d \sK_r, \; s \in [0,T-t], \\[1em]
			\displaystyle \sY \geq \xi^{t,\omega}_{t+\cdot}, \\[1em]
			\displaystyle \int_0^{T-t}(\sY_{r-} - (\xi^{t,\omega}_{t+\cdot})_{r-}) \d \sK_r = 0,
		\end{cases}
	\end{equation}
	such that $\sY = (\sY_s)_{s \in [0,T-t]}$ is right-continuous, $\F_+$--adapted, and of class $(D)$ relative to $(\F_+,\P)$, $\sM = (\sM_s)_{s \in [0,T-t]}$ is a right-continuous, $(\F_+,\P)$-martingale starting at zero, and $\sK = (\sK_s)_{s \in [0,T-t]}$ is a right-continuous and $\P$--a.s. non-decreasing, $\F$-predictable process starting at zero satisfying $\E^\P[\sK_{T-t}] < \infty$. By \cite[Corollary 2.9]{klimsiak2015reflected}, we must have
	\[
	\sY_s = \underset{\tau\in\cT_{\smalltext{s}\smalltext{,}\smalltext{T}\smalltext{-}\smalltext{t}}(\G_\smalltext{+})}{{\esssup}^\P} \E^\P[\xi^{t,\omega}_{t+\tau}|\cG_{s+}]
	= Y^{\P}_s(T-t,\xi^{t,\omega}_{t+\cdot}), \; \textnormal{$\P$--a.s.}, \; s \in [0,T-t].
	\]
	Note that $\xi^{t,\omega}_{t+\cdot}$ is $\F_+$-adapted by Galmarino's test; see \cite[Theorem IV.101.(b)]{dellacherie1978probabilities}.
	
	\medskip
	It follows from \cite[Theorem 2.13]{klimsiak2015reflected} that
	\[
	Y^{\P,n}_0 (T-t,\xi^{t,\omega}_{t+\cdot}) \nearrow Y^{\P}_0 (T-t,\xi^{t,\omega}_{t+\cdot}), \; n \rightarrow \infty, \; \textnormal{$\P$--a.s.},
	\]
	where $Y^{\P,n} (T-t,\xi^{t,\omega}_{t+\cdot})$ denotes the first component of the, up to $\P$--indistinguishability, unique solution $(\sY^n,\sM^n)$ to the BSDE
	\[
	\sY^n_s = \xi^{t,\omega}_{T} + n \int_s^{T-t} \max\{0,\xi^{t,\omega}_r-\sY^n_r\} \d r - \int_s^{T-t} \d\sM^n_r, \; s \in [0,T-t], \; \textnormal{$\P$--a.s.},
	\]
	where $\sY^n = (\sY^n_s)_{s \in [0,T-t]}$ is right-continuous and of Doob class $(D)$ relative to $(\F_+,\P)$, and $\sM^n = (\sM^n_s)_{s \in [0,T-t]}$ is a right-continuous $(\F_+,\P)$-martingale starting at zero. The existence and uniqueness are ensured by \cite[Corollary 2.2 and Theorem 2.7]{klimsiak2013dirichlet}; it is straightforward to check that conditions (H1)--(H4) in \cite{klimsiak2013dirichlet} are satisfied, since, in particular, (H4) holds because
	\[
	\E^\P\bigg[\int_0^{T-t} \max\{0,\xi^{t,\omega}_{t+r}\} \d r\bigg]
	= \int_0^{T-t}\E^\P\big[ \max\{0,\xi^{t,\omega}_{t+r}\} \big] \d r
	\leq (T-t) \sup_{\tau \in \cT_{\smalltext{0}\smalltext{,}\smalltext{T}\smalltext{-}\smalltext{t}}(\G_\smalltext{+})}\E^\P\big[ \max\{0,\xi^{t,\omega}_{t+\tau}\}\big] < \infty.
	\]
	The last inequality follows because $\P$--uniform integrability implies boundedness in $\L^1(\P)$. We now take a closer look at the construction of $Y^{\P,n} (T-t,\xi^{t,\omega}_{t+\cdot})$. From the proof of \cite[Theorem 2.7]{klimsiak2013dirichlet}, we find that
	\[
	\lim_{m \rightarrow \infty}\E^\P\big[Y^{\P,n,m}_0 (T-t,\xi^{t,\omega}_{t+\cdot})\big] = \E^\P\big[Y^{\P,n}_0 (T-t,\xi^{t,\omega}_{t+\cdot})\big],
	\]
	where $Y^{\P,n,m} (T-t,\xi^{t,\omega}_{t+\cdot})$ denotes the first component of the, up to $\P$-indistinguishability, unique solution $(\sY^{n,m},\sM^{n,m})$ in $\cS^2_{T-t}(\F_+,\P) \times \cH^{2}_{0,T-t}(\F_+,\P)$ to the BSDE
	\begin{equation}\label{eq_BSDE_n_m}
		\sY^{n,m}_s = \xi^{t,\omega,m} + \int_s^{T-t} f^{t,\omega,n,m}_r(\sY^{n,m}_r) \d r - \int_s^{T-t} \d \sM^{n,m}_r, \; s \in [0,T-t], \; \textnormal{$\P$--a.s.},
	\end{equation}
	where
	\[
	\xi^{t,\omega,m} \coloneqq \min\{\max\{\xi^{t,\omega}_T,-m\},m\}, \; f^{t,\omega,n,m}_r(y) \coloneqq n\max\{0,\xi^{t,\omega}_r-y\} - n\max\{0,\xi^{t,\omega}_r\} + \min\{\max\{n\max\{0,\xi^{t,\omega}_r\},-m\},m\},
	\]
	and $\cS^2_{T-t}(\F_+,\P)$ denotes the Banach space of real-valued, right-continuous, $\F_+$-adapted processes $Y = (Y_s)_{s\in[0,T-t]}$ satisfying
	\[
	\|Y\|^2_{\cS^\smalltext{2}_{\smalltext{T}\smalltext{-}\smalltext{t}}} 
	\coloneqq \E^\P\bigg[ \sup_{s \in [0,T-t]} |Y_s|^2 \bigg] < \infty,
	\]
	and $\cH^2_{0,T-t}(\F_+,\P)$ the Banach space of real-valued, right-continuous, $(\F_+,\P)$-martingales $M = (M_s)_{s \in [0,T-t]}$ starting at zero and satisfying
	\[
	\|M\|^2_{\cH^\smalltext{2}_{\smalltext{0}\smalltext{,}\smalltext{T}\smalltext{-}\smalltext{t}}}
	\coloneqq \E^\P\big[|M_{T-t}|^2\big] < \infty.
	\]
	The existence and uniqueness of the solution to \eqref{eq_BSDE_n_m} is ensured by \cite[Lemma 2.6]{klimsiak2013dirichlet}. There, one further shows that
	\[
	Y^{\P,n,m,\ell}_s (T-t,\xi^{t,\omega}_{t+\cdot}) \nearrow Y^{\P,n,m}_s (T-t,\xi^{t,\omega}_{t+\cdot}), \; \ell \rightarrow \infty, \; s \in [0,T-t], \; \textnormal{$\P$--a.s.},
	\]
	where $Y^{\P,n,m,\ell} (T-t,\xi^{t,\omega}_{t+\cdot})$ denotes the first component of the unique solution $(\sY^{n,m,\ell},\sM^{n,m,\ell})$ in $\cS^2_{T-t}(\F_+,\P) \times \cH^{2}_{0,T-t}(\F_+,\P)$ to the BSDE
	\[
	\sY^{n,m,\ell}_s = \xi^{t,\omega,m} + \int_s^{T-t} f^{t,\omega,n,m,\ell}_r(\sY^{n,m,\ell}_r) \d r - \int_s^{T-t} \d \sM^{n,m,\ell}_r, \; s \in [0,T-t], \; \textnormal{$\P$--a.s.},
	\]
	where
	\[
	f^{t,\omega,n,m,\ell}_r(y) \coloneqq \inf_{q \in \Q}\big\{\ell |y-q| + f^{t,\omega,n,m}_r(q)\big\};
	\]
	here, $\Q$ denotes the rational numbers. The last BSDE is a standard BSDE with a bounded terminal condition and with a generator satisfying
	\[
	|f^{t,\omega,n,m,\ell}_r(0)| \leq m, \; \textnormal{and} \; |f^{t,\omega,n,m,\ell}_r(y) - f^{t,\omega,n,m,\ell}_r(y^\prime)| \leq \ell |y-y^\prime|.
	\]
	We are thus now in the standard setting of BSDEs with Lipschitz generators. Thus, the construction of the component $Y^{\P,n,m,\ell} (T-t,\xi^{t,\omega}_{t+\cdot})$ follows the usual scheme and is therefore the limit in $\cS^2_{T-t}(\F_+,\P)$ of a Picard iteration
	\[
	Y^{\P,n,m,\ell,1} (T-t,\xi^{t,\omega}_{t+\cdot}), \; Y^{\P,n,m,\ell,2} (T-t,\xi^{t,\omega}_{t+\cdot}), \; Y^{\P,n,m,\ell,3} (T-t,\xi^{t,\omega}_{t+\cdot}), \ldots
	\]
	where we recursively define
	\[
	Y^{\P,n,m,\ell,k+1}_s (T-t,\xi^{t,\omega}_{t+\cdot})
	\coloneqq \E^\P\bigg[ \xi^{t,\omega,m} + \int_s^{T-t} f^{t,\omega,n,m,\ell}_s(Y^{\P,n,m,\ell,k}_s) \d r \bigg| \cF_{s+} \bigg], \; s \in [0,T-t],
	\]
	with starting point $Y^{\P,n,m,\ell,0} (T-t,\xi^{t,\omega}_{t+\cdot}) \equiv 0$. A straightforward modification of \cite[Lemma 3.1 and 3.5]{neufeld2014measurability} then yields the Borel-measurability of
	\[
	(\omega,\P) \longmapsto \E^\P\big[Y^{\P,n,m,\ell,k}_0\big],
	\]
	for each positive integer $n,m,\ell$, and $k$ and a general $\P\in\fP(\Omega)$. Now consider the Borel-measurable (and thus upper semi-analytic) map $\Y : \Omega \times \fP(\Omega) \longrightarrow [-\infty,\infty]$ defined by
	\[
	\Y(\omega,\P) \coloneqq \limsup_{n\rightarrow \infty} \limsup_{m \rightarrow \infty} \limsup_{\ell \rightarrow \infty} \limsup_{k \rightarrow \infty} \E^\P\big[Y^{\P,n,m,\ell,k}_0\big]. 
	\]
	Our previous chain of arguments thus shows that
	\begin{equation}\label{eq_definition_bold_Y}
		\Y(\omega,\P) = \E^\P\big[Y^{\P}_0(T-t,\xi^{t,\omega}_{t+\cdot})\big], \; \P\in\fP(t,\omega).
	\end{equation}
	We therefore conclude that
	\[
	\cY_t(\omega) = \sup_{\P\in\fP(t,\omega)} \E^\P\big[Y^{\P}_0(T-t,\xi^{t,\omega}_{t+\cdot})\big] = \sup_{\P\in\fP(t,\omega)} \Y(\omega,\P), \; \omega \in \Omega,
	\]
	is upper semi-analytic by \cite[Proposition 7.47]{bertsekas1978stochastic}, and thus also $\cF^\ast$-measurable (see \cite[Corollary 8.4.3]{cohn2013measure}). Furthermore, since $\fP(t,\omega) = \fP(t,\omega_{\cdot\land t})$ and $Y^{\P}(T-t,\xi^{t,\omega}_{t+\cdot}) = Y^{\P}(T-t,\xi^{t,\omega_{\smalltext{\cdot}\smalltext{\land}\smalltext{t}}}_{t+\cdot})$, $\P$--a.s., it follows that
	\[
	\cY_t(\omega) = \cY_t(\omega_{\cdot\land t}), \; \omega \in \Omega,
	\] 
	from which we deduce that $\cY_t$ is $\cF^\ast_t$-measurable; see the arguments in the proof of \cite[Lemma 2.5]{nutz2013constructing}.
	
	\medskip
	Before turning to the next step, we need the following intermediate result.
	
	\begin{lemma}\label{lem_conditioning}
		Let $(\omega,t) \in \Omega\times[0,T]$ and $\P\in\fP(t,\omega)$. There exists a $\P$--null set $N \in \cF_t$ with
		\[
		\P^{t,\omega}\in\fP(t,\omega) \; \textnormal{and} \; \E^{\P}[Y^{\P}_t(T,\xi)|\cF_t](\omega)
		= \E^{\P^{\smalltext{t}\smalltext{,}\smalltext{\omega}}}[Y^{\P^{\smalltext{t}\smalltext{,}\smalltext{\omega}}}_0(T-t,\xi^{t,\omega}_{t+\cdot})], \; \omega \in \Omega\setminus N.
		\]
	\end{lemma}
	
	\begin{proof}
		We fix $(\omega,t) \in \Omega \times [0,T]$ and then $\P\in\fP(t,\omega)$. The process $Y^\P(T,\xi)$ is the first component of the solution $(\sY,\sM,\sK)$ to the reflected BSDE with terminal time $T$, obstacle $\xi$ and generator zero, where $\sY = (\sY_s)_{s \in [0,T]}$ is $\P$--a.s. uniquely determined as a right-continuous, $\F_+$-adapted process of class $(D)$ relative to $(\F_+,\P)$, $\sM = (\sM_s)_{s \in [0,T]}$ is a right-continuous $(\F_+,\P)$-martingale starting at zero, and $\sK = (\sK_s)_{s \in [0,T]}$ is a right-continuous and $\P$--a.s. non-decreasing, $\F$-predictable process starting at zero satisfying $\E^\P[\sK_T] < \infty$. It is straightforward to check that 
		\begin{equation*}
			\begin{cases}
				\displaystyle \sY^{t,\omega}_{t+s} = \xi^{t,\omega}_{T} - \int_s^{T-t} \d(\sM^{t,\omega}_{t+\cdot}-\sM^{t,\omega}_t)_r + \int_s^{T-t}\d (\sK^{t,\omega}_{t+\cdot}-\sK^{t,\omega}_t)_r, \; s \in [0,T-t], \\[1em]
				\displaystyle \sY^{t,\omega}_{t+\cdot} \geq \xi^{t,\omega}_{t+\cdot}, \\[1em]
				\displaystyle \int_0^{T-t}(\sY^{t,\omega}_{(t+r)-} - \xi^{t,\omega}_{(t+r)-}) \d (\sK^{t,\omega}_{t+\cdot}-\sK^{t,\omega}_t)_r = 0,
			\end{cases}
		\end{equation*}
		holds $\P^{t,\omega}$--a.s., for $\P$--a.e. $\omega \in \Omega$. Suppose for the moment that we have shown that
		\begin{enumerate}
			\item[(i)] $\sM^{t,\omega}_{t+\cdot}-\sM^{t,\omega}_t$ is a right-continuous $(\F_+,\P^{t,\omega})$-martingale on $[0,T-t]$,
			\item[(ii)] $\sK^{t,\omega}_{t+\cdot}-\sK^{t,\omega}_t$ is right-continuous, $\P$--a.s. non-decreasing, $\F$-predictable, and satisfies $\E^{\P^{\smalltext{t}\smalltext{,}\smalltext{\omega}}}\big[\sK^{t,\omega}_{T}-\sK^{t,\omega}_t\big]<\infty$, and
			\item[(iii)] $\E^{\P^{\smalltext{t}\smalltext{,}\smalltext{\omega}}}\big[|\sY^{t,\omega}_t|\big]<\infty$, for $\P$--a.e. $\omega \in\Omega$.
		\end{enumerate}
		Since
		\[
		\sY^{t,\omega}_{t+s} = \sY^{t,\omega}_t + (\sM^{t,\omega}_{t+s}-\sM^{t,\omega}_t) - (\sK^{t,\omega}_{t+s}-\sK^{t,\omega}_t), \; s \in [0,T-t], \; \textnormal{$\P^{t,\omega}$--a.s.},
		\]
		it follows that the family
		\[
		\big\{\sY^{t,\omega}_{t+\tau} \,\big|\, \tau \in \cT_{0,T-t}(\F_+)\big\}
		\]
		is $\P^{t,\omega}$--uniformly integrable for $\P$--a.e. $\omega \in\Omega$. Hence, $(\sY^{t,\omega}_{t+\cdot},\sM^{t,\omega}_{t+s}-\sM^{t,\omega}_t,\sK^{t,\omega}_{t+\cdot}-\sK^{t,\omega}_t)$ is the unique solution to the reflected BSDE relative to $(\F_+,\P^{t,\omega})$ whose generator is identically zero, terminal time is $T-t$, and obstacle is $\xi^{t,\omega}_{t+\cdot}$, for $\P$--a.e. $\omega \in \Omega$. In particular, it then follows from \cite[Corollary 2.9 and Theorem 2.13]{klimsiak2015reflected} that
		\[
		\sY^{t,\omega}_{t+s} =  \underset{\tau\in\cT_{\smalltext{s}\smalltext{,}\smalltext{T}\smalltext{-}\smalltext{t}}(\G_\smalltext{+})}{{\esssup}^{\P^{\smalltext{t}\smalltext{,}\smalltext{\omega}}}} \E^{\P^{\smalltext{t}\smalltext{,}\smalltext{\omega}}}[\xi^{t,\omega}_{t+\tau}|\cG_{s+}] = Y^{\P^{\smalltext{t}\smalltext{,}\smalltext{\omega}}}_s(T-t,\xi^{t,\omega}_{t+\cdot}), \; \textnormal{$\P^{t,\omega}$--a.s.}, \; s \in [0,T-t], \; \textnormal{$\P$--a.e. $\omega \in \Omega$.}
		\]
		This then yields
		\[
		\E^\P[Y^\P_t(T,\xi)|\cF_t](\omega) = \E^\P[\sY_t |\cF_t](\omega) = \E^{\P^{\smalltext{t}\smalltext{,}\smalltext{\omega}}}[(\sY_t)^{t,\omega}] = \E^{\P^{\smalltext{t}\smalltext{,}\smalltext{\omega}}}[\sY_t^{t,\omega}] = \E^{\P^{\smalltext{t}\smalltext{,}\smalltext{\omega}}}[Y^{\P^{\smalltext{t}\smalltext{,}\smalltext{\omega}}}_0(T-t,\xi^{t,\omega}_{t+\cdot})], \; \textnormal{$\P$--a.e. $\omega \in \Omega$.}
		\]
		We now show that the $\P$--exceptional set on which the outer left-hand and right-hand sides do not coincide can be chosen to be in $\cF_t$.
		
		\medskip
		The map $\Y : \Omega \times \fP(\Omega) \longrightarrow [-\infty,\infty]$ constructed in \eqref{eq_definition_bold_Y} is a Borel-measurable extension of
		\[ 
		\{(\omega^\prime,\P^\prime) \,|\, \omega^\prime \in \Omega, \; \P^\prime\in\fP(t,\omega)\} \ni (\omega^\prime,\P^\prime) \longmapsto \E^{\P^\smalltext{\prime}}[Y^{\P^\smalltext{\prime}}_0(T-t,\xi^{t,\omega^\smalltext{\prime}}_{t+\cdot})] \in \R,
		\]
		Moreover, the map
		\[
		\Phi : \Omega \ni \omega \longmapsto (\omega_{\cdot\land t},\P^{t,\omega_{\smalltext{\cdot}\smalltext{\land}\smalltext{t}}}) \in \Omega \times \fP(\Omega)
		\]
		is $\cF_{t}$-measurable by \cite[Proposition~7.25, p.~133]{bertsekas1978stochastic} and Galmarino's test. Note that we have $\P^{t,\omega_{\cdot \land t}} = \P^{t,\omega}$ identically. By {\rm\Cref{ass::measures}}.$(ii)$, the map $\Phi$ takes values $\P$--a.s. in the analytic set $\{(\omega^\prime,\P^\prime) \,|\, \omega^\prime \in \Omega, \; \P^\prime \in \fP(t,\omega)\}$. Since analytic sets are universally measurable (see \cite[Corollary~8.4.3]{cohn2013measure}), the preimage of $\{(\omega^\prime,\P^\prime) \,|\, \omega^\prime \in \Omega, \; \P^\prime \in \fP(t,\omega)\}$ under $\Phi$ is therefore $\cF^\ast_{t}$-measurable. Consequently, there exists a set $\Omega_1 \in \cF_{t}$ with $\P[\Omega_1] = 1$ on which $\Phi$ attains values in $\{(\omega^\prime,\P^\prime) : \omega^\prime \in \Omega, \; \P^\prime \in \fP(t,\omega)\}$. In particular, this yields
		\[
		(\Y \circ \Phi)(\omega)
		= \E^{\P^{\smalltext{t}\smalltext{,}\smalltext{\omega}}}\big[Y^{\P^{\smalltext{t}\smalltext{,}\smalltext{\omega}}}_0 (T-t,\xi^{t,\omega}_{t+\cdot})\big], \omega \in \Omega_1,
		\]
		and then
		\[
		\E^\P\big[Y^{\P}_{t}(T,\xi) \big| \cF_{t}\big]
		= (\Y \circ \Phi), \; \textnormal{$\P$--a.s.}
		\]
		Since both sides are $\cF^\ast_{t}$-measurable, there exists a further set $\Omega_2 \in \cF_t$ with $\P[\Omega_2] = 1$ inside which equality holds above. For $\omega \in \Omega_1\cap\Omega_2$, we then have
		\[
		\E^\P\big[Y^{\P}_{t}(T,\xi) \big| \cF_{t}\big](\omega)
		= (\Y\circ\Phi)(\omega) 
		= \E^{\P^{\smalltext{t}\smalltext{,}\smalltext{\omega}}}\big[Y^{\P^{\smalltext{t}\smalltext{,}\smalltext{\omega}}}_0 (T-t,\xi^{t,\omega}_{t+\cdot})\big].		
		\]
		Thus, the set $N \coloneqq \Omega\setminus(\Omega_1\cap\Omega_2) \in \cF_{t}$ satisfies the desired properties.

		\medskip
		It remains to argue that (i), (ii) and (iii) hold. That (iii) holds follows simply from
		\[
		\E^\P\Big[\E^{\P^{\smalltext{t}\smalltext{,}\smalltext{\cdot}}}\big[|\sY^{t,\cdot}_t|\big]\Big] = \E^\P\Big[\E^{\P}\big[|\sY_t| \big|\cF_t\big]\Big] = \E^{\P}\big[|\sY_t|\big] < \infty.
		\]
		The integrability in (ii) follows analogously, right-continuity is immediate, and the property of being $\P^{t,\omega}$--a.s. non-decreasing for $\P$--a.e. $\omega \in \Omega$ can be shown using the right-continuity of the paths. That the process in (ii) is $\F$-predictable follows from product-measurability and that it is adapted to $\F_-$; see \cite[Theorem IV.97.(b) and IV.99.(b)]{dellacherie1978probabilities}. Lastly, we prove (i). We take inspiration from the proof of \cite[Lemma 3.3]{neufeld2016nonlinear}. It follows from Galmarino's test (see \cite[Theorem IV.101.(b)]{dellacherie1978probabilities}) that $\sM^{t,\omega}_{t+\cdot}$ is $\F_+$-adapted for every $\omega \in \Omega$. Next, let $\D$ be a dense subset in $[0,T-t]$ that contains $T-t$. It follows from
		\begin{equation*}
			\E^\P\Big[\E^{\P^{\smalltext{t}\smalltext{,}\smalltext{\cdot}}}\big[|\sM^{t,\cdot}_{t+s}|\big]\Big] = \E^\P\big[|\sM_{t+s}|\big] < \infty, \; s \in \D,
		\end{equation*}
		that 
		\[
		\E^{\P^{\smalltext{t}\smalltext{,}\smalltext{\omega}}}\big[|\sM^{t,\omega}_{t+s}|\big] < \infty, \; s \in \D, \; \textnormal{$\P$--a.e. $\omega \in \Omega$.}
		\]
		We turn to the martingale property along times in $\D$ for the moment.
		Fix numbers $r < r_n < s$ in $\D$ with $r_n \downarrow r$, a bounded, $\cF_{r_\smalltext{n}}$-measurable function $g_n$, and define $\tilde{g}_n(\omega)\coloneqq g_n(\omega_{t+\cdot}-\omega_{t})$. Then $\tilde{g}^{t,\omega}_n = g_n$ and $\tilde{g}_n$ is $\cF_{t+r_\smalltext{n}}$-measurable. By the optional sampling theorem under $\P$ (see \cite[Corollary 3.2.8]{weizsaecker1990stochastic}), we have
		\begin{equation*}
			\E^{\P^{\smalltext{t}\smalltext{,}\smalltext{\omega}}}\big[ \big(\sM^{t,\omega}_{t+s}-\sM^{t,\omega}_{t+r_\smalltext{n}} \big) g_n\big] 
			= \E^{\P}\big[ (\sM_{t+s}-\sM_{t+r_\smalltext{n}} ) \tilde{g}_n \big| \cF_{t+}\big](\omega)
			= \E^{\P}\Big[ \E^\P\big[ (\sM_{t+s}-\sM_{t+r_\smalltext{n}} ) \big| \cF_{(t+r_\smalltext{n})+} \big] \tilde{g}_n \Big| \cF_{t+}\Big](\omega) = 0,
		\end{equation*}
		for $\P$--a.e. $\omega \in \Omega$. A functional monotone class argument, together with the separability of $\cF_{r_\smalltext{n}}$, then implies that
		\begin{equation*}
			\E^{\P^{\smalltext{t}\smalltext{,}\smalltext{\cdot}}}\big[ \big(\sM^{t,\cdot}_{t+s}-\sM^{t,\cdot}_{t+r_\smalltext{n}} ) g\big]=0,
		\end{equation*}
		holds simultaneously for all bounded and $\cF_{r_\smalltext{n}}$-measurable functions $g$, on the complement of a $\P$--null set. This implies that
		\begin{equation*}
			\E^{\P^{\smalltext{t}\smalltext{,}\smalltext{\omega}}}\big[ \sM^{t,\omega}_{t+s}\big|\cF_{r_\smalltext{n}}\big] 
			= \E^{\P^{\smalltext{t}\smalltext{,}\smalltext{\omega}}}\big[\sM^{t,\omega}_{t+r_\smalltext{n}}\big|\cF_{r_\smalltext{n}}\big], \; \textnormal{$\P^{t,\omega}$--a.s.}, \; \textnormal{$\P$--a.e. $\omega \in \Omega$}.
		\end{equation*}
		By the backward martingale convergence theorem \cite[Theorem V.33]{dellacherie1982probabilities} the left-hand side converges $\P^{t,\omega}$--a.s. and in $\L^1(\P^{t,\omega})$ to $\E^{\P^{\smalltext{t}\smalltext{,}\smalltext{\omega}}}[\sM^{t,\omega}_{t+s}|\cF_{r+}]$ for $\P$--a.e. $\omega \in \Omega$. On the other hand, for the right-hand side, we obtain from
		\[
		\E^{\P^{\smalltext{t}\smalltext{,}\smalltext{\omega}}}\Big[\big|\sM^{t,\omega}_{t+r} - \E^{\P^{\smalltext{t}\smalltext{,}\smalltext{\omega}}}\big[\sM^{t,\omega}_{t+r_\smalltext{n}}\big|\cF_{r_\smalltext{n}}\big]\big| \Big] 
		\leq \E^{\P^{\smalltext{t}\smalltext{,}\smalltext{\omega}}}\Big[\big|\sM^{t,\omega}_{t+r} - \sM^{t,\omega}_{t+r_\smalltext{n}}\big| \Big] = \E^\P\Big[\big|\sM_{t+r} - \sM_{t+r_\smalltext{n}}\big|\Big|\cF_t\Big](\omega), \; \textnormal{$\P$--a.s.},
		\]
		that
		\[
		\E^\P\bigg[\E^{\P^{\smalltext{t}\smalltext{,}\smalltext{\cdot}}}\Big[\big|\sM^{t,\cdot}_{t+r} - \E^{\P^{\smalltext{t}\smalltext{,}\smalltext{\cdot}}}\big[\sM^{t,\cdot}_{t+r_\smalltext{n}}\big|\cF_{r_\smalltext{n}}\big]\big| \Big]\bigg] 
		\leq \E^\P \Big[\big|\sM_{t+r} - \sM_{t+r_\smalltext{n}}\big|\Big] \xrightarrow{n\rightarrow\infty} 0.
		\]
		Here, we used the fact that the martingale $\sM$ is $\P$--uniformly integrable and has right-continuous paths. This then implies that, along a subsequence (depending only on $\P$), if necessary,
		\[
		\lim_{n \rightarrow \infty}\E^{\P^{\smalltext{t}\smalltext{,}\smalltext{\omega}}}\big[\sM^{t,\omega}_{t+r_\smalltext{n}}\big|\cF_{r_\smalltext{n}}\big]
		= \sM^{t,\omega}_{t+r},
		\]
		in $\L^1(\P^{t,\omega})$, for $\P$--a.e. $\omega \in \Omega$. Therefore,
		\[
		\E^{\P^{\smalltext{t}\smalltext{,}\smalltext{\omega}}}[\sM^{t,\omega}_{t+s}|\cF_{r+}] = \sM^{t,\omega}_{t+r}, \; \textnormal{$\P^{t,\omega}$--a.s.}, \; \textnormal{$\P$--a.e. $\omega \in \Omega$.}
		\]
		Thus, $\sM^{t,\omega}_{t+\cdot}$ is an $(\F_+,\P^{t,\omega})$-martingale on $\D$ for $\P$--a.e. $\omega \in \Omega$. The integrability and martingale property extend to $[0,T]$ by right-continuity and the backward martingale convergence theorem. This concludes the proof of this intermediate result.
	\end{proof}
	
	\medskip
	{\bf Step 2:} We prove that
	\[
	\cY_t = \underset{\bar{\P}\in\fP(\cF_t,\P)}{{\esssup}^\P} \E^{\bar{\P}}\big[Y^{\bar{\P}}_t(T,\xi)\big|\cF_t\big], \; \textnormal{$\P$--a.s.}, \; t \in [0,T], \; \P\in\fP,
	\]
	where $\fP(\cF_t,\P) \coloneqq \big\{\overline{\P}\in\fP \,|\, \textnormal{$\overline{\P} = \P$ on $\cF_t$}\big\}$. We fix $t \in [0,T]$, $\P\in\fP$, and then $\overline{\P}\in\fP(\cF_t,\P)$. It follows from \Cref{lem_conditioning} that
	\[
	\E^{\bar{\P}}\big[Y^{\bar{\P}}_t(T,\xi) \big| \cF_t\big](\omega) 
	= \E^{{\bar{\P}}^{\smalltext{t}\smalltext{,}\smalltext{\omega}}} [Y^{{\bar{\P}}^{\smalltext{t}\smalltext{,}\smalltext{\omega}}}_0(T-t,\xi^{t,\omega}_{t+\cdot})], 
	\; \textnormal{${\overline{\P}}$--a.e. $\omega \in \Omega$.}
	\]
	Since $\bar{\P}^{t,\omega} \in \fP(t,\omega)$ for ${\overline{\P}}$--a.e. $\omega \in \Omega$ by \Cref{ass::measures}.$(ii)$, we find
	\[
	\E^{\bar{\P}}\big[Y^{\bar{\P}}_t(T,\xi) \big| \cF_t\big](\omega) 
	= \E^{{\bar{\P}}^{\smalltext{t}\smalltext{,}\smalltext{\omega}}} [Y^{t,\omega,{\bar{\P}}^{\smalltext{t}\smalltext{,}\smalltext{\omega}}}_0(T-t,\xi^{t,\omega}_{t+\cdot})] 
	\leq \sup_{\P^{\smalltext{\prime}} \in \fP(t,\omega)} \E^{\P^\smalltext{\prime}}\big[Y^{\P^\smalltext{\prime}}_0(T-t,\xi^{t,\omega}_{t+\cdot})\big] = \cY_t(\omega), \; \textnormal{for $\overline{\P}$--a.e. $\omega \in \Omega$.}
	\]
	Since the left-hand side is $\cF_t$-measurable, the right-hand side is $\cF^\ast_t$-measurble, and $\overline{\P}$ agrees with $\P$ on $\cF_t$ and thus also on $\cF^\ast_t$, we obtain
	\[
	\E^{\bar{\P}}\big[Y^{\bar{\P}}_t(T,\xi) \big| \cF_t\big] \leq \cY_t, \; \textnormal{$\P$--a.s.}
	\]
	This then yields
	\[
	\underset{\overline{\P}\in\fP(\cF_t,\P)}{{\esssup}^\P} \E^{\bar{\P}}\big[Y^{\bar{\P}}_t(T,\xi)\big|\cF_t\big]
	\leq \cY_t, \; \textnormal{$\P$--a.s.}, \; t \in [0,T], \; \P\in\fP.
	\]
	
	\medskip
	We establish the other inequality using a selection theorem for upper semi-analytic functions. Fix $\P \in \fP$, $t \in [0,T]$, and a positive number $\varepsilon$. By \cite[Proposition 7.50]{bertsekas1978stochastic}, there exists a universally measurable map $\Q^\prime : \Omega \longrightarrow \fP(\Omega)$ such that 
	\begin{equation}\label{eq::analytic_selection2}
		\Q^\prime(\omega) \in \fP(t,\omega),
		\;
		\E^{\Q^\smalltext{\prime}(\omega)} \big[Y^{{\Q^\smalltext{\prime}(\omega)}}_0(T-t,\xi^{t,\omega}_{t+\cdot})\big] 
		\geq \big(\cY_t(\omega) - \varepsilon\big) \1_{\{\cY_\smalltext{t} < \infty\}}(\omega) + \frac{1}{\varepsilon}\1_{\{\cY_\smalltext{t} = \infty\}}(\omega),
	\end{equation}
	for each $\omega \in \Omega$ for which $\fP(t,\omega) \neq \varnothing$. The map $\widetilde\Q : \Omega \longrightarrow \fP(\Omega)$ defined by $\widetilde\Q(\omega) \coloneqq \Q^\prime({\omega_{\cdot \land t}})$ is $\cF^\ast_t$-measurable and also satisfies \eqref{eq::analytic_selection2} for each $\omega \in \Omega$ with $\fP(t,\omega)\neq \varnothing$. According to \cite[Lemma 1.27]{kallenberg2021foundations}, there exists an $\cF_t$-measurable map $\Q : \Omega \longrightarrow \fP(\Omega)$ satisfying $\Q(\cdot) = \widetilde\Q(\cdot)$ outside a $\P$--null set. Since $\fP(t,\omega) \neq \varnothing$ for $\P$--a.e. $\omega \in\Omega$ by \Cref{ass::measures}.$(ii)$, we have that $\Q(\omega)$ satisfies \eqref{eq::analytic_selection2} for $\P$--a.e. $\omega \in \Omega$. Let
	\begin{equation*}
		\overline{\P}[A] \coloneqq  \int_\Omega\int_\Omega \big(\1_A\big)^{t,\omega}(\omega^\prime)\Q(\omega;\d\omega^\prime)\P(\d\omega), \; A \in \cF.
	\end{equation*}
	Then, $\overline\P \in \fP$ by \Cref{ass::measures}.$(iii)$, $\overline{\P} = \P$ on $\cF_t$, and $\Q(\omega) = \bar\P^{t,\omega}$ for $\overline{\P}$--a.e. $\omega \in \Omega$ and then $\P$--a.e. $\omega \in \Omega$; the map $\omega \longmapsto \bar{\P}^{t,\omega}$ is $\cF_t$-measurable by \cite[Proposition~7.25, p.~133]{bertsekas1978stochastic} and Galmarino's test since $\bar{\P}^{t,\omega} = \bar{\P}^{t,\omega_{\smalltext{\cdot}\smalltext{\land}\smalltext{t}}}$. We obtain from \Cref{lem_conditioning} that
	\begin{align*}
		\E^{\bar\P}\big[Y^{\bar\P}_t(T,\xi)\big|\cF_t\big](\omega) 
		&= \E^{\bar{\P}^{\smalltext{t}\smalltext{,}\smalltext{\omega}}}\big[Y^{\bar{\P}^{\smalltext{t}\smalltext{,}\smalltext{\omega}}}_0(T-t,\xi^{t,\omega}_{t+\cdot})\big] \\
		&= \E^{\Q(\omega)}\big[Y^{\Q(\omega)}_0(T-t,\xi^{t,\omega}_{t+\cdot})\big]
		\geq 
		\big(\cY_t(\omega) - \varepsilon\big) \1_{\{\cY_\smalltext{t} < \infty\}}(\omega) + \frac{1}{\varepsilon}\1_{\{\cY_\smalltext{t} = \infty\}}(\omega), \; \textnormal{$\P$--a.e. $\omega \in \Omega$,}
	\end{align*}
	and then
	\begin{equation*}
		\underset{\P^\prime \in \fP(\cF_\smalltext{t},\P)}{{\esssup}^\P} \E^{\P^\smalltext{\prime}} \big[ Y^{\P^\smalltext{\prime}}_t(T,\xi)\big| \cF_t\big] 
		\geq 
		\big(\cY_t - \varepsilon\big) \1_{\{\cY_\smalltext{t} < \infty\}} + \frac{1}{\varepsilon}\1_{\{\cY_\smalltext{t} = \infty\}}, \; \text{$\P$--a.s.}
	\end{equation*}
	Since $\varepsilon$ was arbitrary, we obtain the desired inequality.
	
	\medskip
	{\bf Step 3:} We show that $\cY$ is an $(\F^\ast,\P)$--super-martingale on $[0,T]$ for every $\P \in \fP$; here $\F^\ast \coloneqq (\cF^\ast_t)_{t \in [0,T]}$, where $\cF^\ast_t$ is the universal completion of $\cF_t$. We fix $\P \in \fP$ and first establish that, for $t \in [0,T]$, the family
	\begin{equation}\label{eq_upward_directed}
		\big\{\E^{\bar{\P}}[Y^{\bar{\P}}_t(T,\xi) | \cF_t] \,\big|\, \overline{\P} \in \fP(\cF_t,\P) \big\}
	\end{equation}
	is $\P$--upward directed. For $\P_1$ and $\P_2$ in $\fP(\cF_t,\P)$, we consider the set
	\[
	B \coloneqq \big\{\E^{\P_\smalltext{1}}[Y^{\P_\smalltext{1}}_t(T,\xi) | \cF_t] \geq \E^{\P_\smalltext{2}}[Y^{\P_\smalltext{2}}_t(T,\xi) | \cF_t] \big\} \subseteq \Omega,
	\]
	and define
	\begin{equation*}
		\Q(\omega;A) \coloneqq\P^{t,\omega}_1[A]\1_{B}(\omega) + \P^{t,\omega}_2[A]\1_{B^\smalltext{c}}(\omega), \; \omega \in \Omega.
	\end{equation*}
	By \Cref{ass::measures}.$(ii)$, we have that $\P^{t,\omega}_1 \in \fP(t,\omega)$ for $\P_1$--a.e. $\omega \in \Omega$. As before, $\omega \longmapsto \P^{t,\omega}_1$ is $\cF_t$-measurable. Then
	\[
	D \coloneqq \big\{\omega \in \Omega \,\big|\, (\omega,\P^{t,\omega}_1) \in \Omega \times \fP(t,\omega)\big\} = \big\{\omega \in \Omega \,\big|\, (\omega_{\cdot\land t},\P^{t,\omega_{\smalltext{\cdot}\smalltext{\land}\smalltext{t}}}_1) \in \Omega \times \fP(t,\omega)\big\}
	\]
	is $\cF^\ast_t$-measurable since it is the preimage of an analytic set in $\Omega\times\fP(\Omega)$ by \Cref{ass::measures}.$(i)$ under the $\cF_t$-measurable map
	\[
	\Omega \ni \omega \longmapsto (\omega_{\cdot\land t},\P^{t,\omega}_1) \in \Omega \times \fP(\Omega).
	\]
	Since $\P$ agrees with $\P_1$ on $\cF_t$, and thus on $\cF^\ast_t$, this yields $\P[D] = \P_1[D] = 1$. In particular, this implies that $\P^{t,\omega}_1 \in \fP(t,\omega)$ for $\P$--a.e. $\omega \in \Omega$. The analogous argument also applies to $\P_2$, and thus $\P^{t,\omega}_2 \in \fP(t,\omega)$ for $\P$--a.e. $\omega \in \Omega$. Then $\Q(\omega;\d\omega^\prime)$ is a kernel on $(\Omega,\cF)$ given $(\Omega,\cF_{t})$ satisfying $\Q(\omega;\d\omega^\prime) \in \fP(t,\omega)$ for $\P$--a.e. $\omega \in \Omega$.
	The probability measure
	\begin{align*}
		\overline{\P}[A] \coloneqq & \int_\Omega\int_\Omega \big(\1_A\big)^{t,\omega}(\omega^\prime)\Q(\omega;\d\omega^\prime)\P(\d\omega)
		= \E^\P[\P_1[A|\cF_{t}]\1_B + \P_2[A|\cF_{t}]\1_{B^\smalltext{c}}] = \P_1[A \cap B] + \P_2[A \cap B^c], \; A \in \cF,
	\end{align*}
	is therefore an element of $\fP$ by \Cref{ass::measures}.$(iii)$. Since $\overline{\P}$ agrees with $\P$ on $\cF_{t}$ and $\bar{\P}^{t,\omega}(\d\omega^\prime) = \Q(\omega,\d\omega^\prime)$ for $\overline{\P}$--a.e. $\omega \in \Omega$ and then $\P$--a.e. $\omega \in \Omega$, we obtain with \Cref{lem_conditioning} that
	\begin{align*}
		\E^{\bar{\P}}\big[Y^{\bar{\P}}_{t}(T,\xi)\big|\cF_{t}\big](\omega)
		&= \E^{\Q(\omega)}[Y^{\Q(\omega)}_{0}(T-t,\xi^{t,\omega}_{t+\cdot})] \\
		&= \E^{\Q(\omega)}[Y^{\Q(\omega)}_{0}(T-t,\xi^{t,\omega}_{t+\cdot})]\1_B(\omega) + \E^{\Q(\omega)}[Y^{\Q(\omega)}_{0}(T-t,\xi^{t,\omega}_{t+\cdot})]\1_{B^\smalltext{c}}(\omega) \\
		&= \E^{\P^{\smalltext{t}\smalltext{,}\smalltext{\omega}}_\smalltext{1}}[Y^{\P^{\smalltext{t}\smalltext{,}\smalltext{\omega}}_\smalltext{1}}_{0}(T-t,\xi^{t,\omega}_{t+\cdot})]\1_B(\omega) + \E^{\P^{\smalltext{t}\smalltext{,}\smalltext{\omega}}_\smalltext{2}}[Y^{\P^{\smalltext{t}\smalltext{,}\smalltext{\omega}}_\smalltext{2}}_{0}(T-t,\xi^{t,\omega}_{t+\cdot})]\1_{B^\smalltext{c}}(\omega) \\
		&= \E^{\P_\smalltext{1}}[Y^{\P_\smalltext{1}}_{t}(T,\xi)|\cF_{t}](\omega)\1_B(\omega) + \E^{\P_\smalltext{2}}[Y^{\P_\smalltext{2}}_{t}(T,\xi)|\cF_{t}](\omega)\1_{B^\smalltext{c}}(\omega) \\
		&\geq \max\big\{\E^{\P_\smalltext{1}}[Y^{\P_\smalltext{1}}_{t}(T,\xi)|\cF_{t}](\omega),\E^{\P_\smalltext{2}}[Y^{\P_\smalltext{2}}_{t}(T,\xi)|\cF_{t}](\omega)\big\}, \; \text{for $\P$--a.e. $\omega \in \Omega$.}
	\end{align*}
	Since $\E^{\bar{\P}}\big[Y^{\bar{\P}}_{t}(T,\xi)\big|\cF_{t}\big]$ is an element of the family in \eqref{eq_upward_directed}, this implies that the family is $\P$--upward directed.
	
	\medskip
	We turn to the super-martingale property. Let $0 \leq s \leq t \leq T$, and let $(\P_n)_{n\in\N} \subseteq \fP(\cF_t,\P)$ be such that
	\[
	\E^{\P_\smalltext{n}}[Y^{\P_\smalltext{n}}_t(T,\xi)|\cF_t] \nearrow \cY_t, \; n \longrightarrow \infty, \; \textnormal{$\P$--a.s.};
	\]
	this is where we use the $\P$--upward directedness of the family in \eqref{eq_upward_directed}; see \cite[Proposition VI-1-1, p.~121]{neveu1975discrete}.
	Since $(\P_n)_{n \in \N}\subseteq \fP(\cF_s,\P)$, we find
	\begin{align*}
		\E^\P[\cY_t|\cF_s] 
		&= \lim_{n \rightarrow\infty} \E^\P\big[ \E^{\P_\smalltext{n}}[Y^{\P_\smalltext{n}}_t(T,\xi)|\cF_t] \big|\cF_s\big] = \lim_{n \rightarrow\infty} \E^{\P_\smalltext{n}}\big[ \E^{\P_\smalltext{n}}[Y^{\P_\smalltext{n}}_t(T,\xi)|\cF_t] \big|\cF_s\big] = \lim_{n \rightarrow\infty} \E^{\P_\smalltext{n}}[Y^{\P_\smalltext{n}}_t(T,\xi)|\cF_s] \\
		&= \lim_{n \rightarrow\infty} \E^{\P_\smalltext{n}}\big[\E^{\P_\smalltext{n}}[Y^{\P_\smalltext{n}}_t(T,\xi)|\cF_{s+}]\big|\cF_s\big] 
		\leq \liminf_{n \rightarrow\infty} \E^{\P_\smalltext{n}}[Y^{\P_\smalltext{n}}_s(T,\xi)\big|\cF_s\big] 
		\leq \underset{\P^\smalltext{\prime}\in\fP(\cF_\smalltext{s},\P)}{{\esssup}^{\P}} \E^{\P^\smalltext{\prime}}[Y^{\P^\smalltext{\prime}}_s(T,\xi)|\cF_s]
		= \cY_s, \; \textnormal{$\P$--a.s.} 
	\end{align*}
	Here we used the fact that $Y^{\P_\smalltext{n}}(T,\xi)$ is an $(\F_{+},\P_n)$--super-martingale. This implies $\E^\P[\cY_t|\cF^\ast_s] \leq \cY_s$, $\P$--a.s., which is the desired martingale property since $\cY$ is $\F^\ast$-adapted. The required integrability can be argued as follows. First, we have for any $\P\in\fP$ that
	\[
	- \E^{\P}\big[|\xi_T|\big|\cF_t\big] \leq \cY_t 
	\leq \underset{\bar{\P}\in\fP(\cF_\smalltext{t},\P)}{{\esssup}^{\P}} \E^{\bar{\P}}[Y^{\bar{\P}}_t(T,\xi^+)|\cF_t], \; \textnormal{$\P$--a.s.},
	\]
	where $\xi^+ \coloneqq \max\{0,\xi\}$. Similarly to before, we can find $(\overline{\P}_n)_{n \in \N} \subseteq \fP(\cF_t,\P)$ such that
	\[
	\E^{\bar{\P}_\smalltext{n}}[Y^{\bar{\P}_\smalltext{n}}_t(T,\xi^+)|\cF_t]
	\nearrow
	\underset{\bar{\P}\in\fP(\cF_\smalltext{t},\P)}{{\esssup}^{\P}} \E^{\bar{\P}}[Y^{\bar{\P}}_t(T,\xi^+)|\cF_t], \; n \longrightarrow \infty, \; \textnormal{$\P$--a.s.},
	\]
	which yields
	\begin{align}\label{eq_integrability_Y_0}
		\E^\P\big[|\cY_t|\big] 
		&\leq \E^{\P}\big[|\xi_T|\big] + \E^\P\Bigg[ \underset{\bar{\P}\in\fP(\cF_\smalltext{t},\P)}{{\esssup}^{\P}} \E^{\bar{\P}}[Y^{\bar{\P}}_t(T,\xi^+)|\cF_t] \Bigg] \nonumber\\
		&= \E^{\P}\big[|\xi_T|\big] + \lim_{n \rightarrow\infty}\E^\P\Big[\E^{\bar{\P}_\smalltext{n}}[Y^{\bar{\P}_\smalltext{n}}_t(T,\xi^+)|\cF_t] \Big] \nonumber\\
		&= \E^{\P}\big[|\xi_T|\big] + \lim_{n \rightarrow\infty}\E^{\bar{\P}_\smalltext{n}}\big[Y^{\bar{\P}_\smalltext{n}}_t(T,\xi^+)\big] \nonumber\\
		&\leq  \E^{\P}\big[|\xi_T|\big] + \sup_{\P^\smalltext{\prime}\in\fP}\E^{\P^\smalltext{\prime}}[Y^{\P^\smalltext{\prime}}_t(T,\xi^+)] \nonumber\\
		&= \E^{\P}\big[|\xi_T|\big] + \sup_{\P^\smalltext{\prime}\in\fP}\E^{\P^\smalltext{\prime}}\bigg[\underset{\tau\in\cT_{\smalltext{t}\smalltext{,}\smalltext{T}}(\G_\smalltext{+})}{{\esssup}^{\P^\smalltext{\prime}}} \E^{\P^\smalltext{\prime}}[\xi^+_\tau|\cG_{t+}]\bigg] 
		\leq \E^{\P}\big[|\xi_T|\big] +  \sup_{\P^\smalltext{\prime}\in\fP}\sup_{\tau\in\cT_{\smalltext{t}\smalltext{,}\smalltext{T}}(\G_\smalltext{+})}\E^{\P^\smalltext{\prime}}[\xi^+_\tau] < \infty.
	\end{align}
	The last inequality follows from the (classical) fact that the family of conditional expectations
	\[
	\big\{ \E^{\P^\smalltext{\prime}}[\xi^+_\tau|\cG_{t+}] \,\big|\, \tau \in \cT_{t,T}(\G_+) \big\}
	\]
	is $\P^\prime$--upward directed. We have shown that $\cY$ is an $(\F^\ast,\P)$--super-martingale on $[0,T]$ for each $\P\in\fP$.
	
	\medskip
	{\bf Step 4:} Define $Y^0 = (Y^0_t)_{t \in [0,T]}$ by
	\begin{equation}\label{def_Y0}
		Y^0_t \coloneqq \limsup_{\D_\smallertext{+} \ni s \downarrow\downarrow t} \cY_t, \; t \in [0,T), \; \textnormal{and} \; Y^0_T \coloneqq \xi_T,
	\end{equation}
	where $\D_+$ denotes the nonnegative dyadic numbers.
	Then $Y^0$ is $\F^\ast_+$-adapted, and the collection $N\subseteq \Omega$ of all $\omega \in\Omega$ where $t \longmapsto Y^0_t(\omega)$ is not a real-valued process with c\`adl\`ag paths is a $\P$--null set for every $\P\in\fP$ by \cite[Theorem VI.2]{dellacherie1982probabilities}, and therefore $N \in \sN^{\fP}$. Moreover, the same result yields that $Y^0$ is an $(\F^\ast_+,\P)$--super-maringale, and that the limit superior in \eqref{def_Y0} is actually a limit outisde a $\P$--null set for every $\P\in\fP$. Therefore, $Y \coloneqq Y^0 \1_{\Omega\setminus N}$ is a real-valued, c\`adl\`ag, $\G_+$-adapted, $(\G_+,\P)$--super-martingale for every $\P\in\fP$.
	
	\medskip
	We show that $Y = (Y_t)_{t \in [0,T]}$ satisfies
	\[
	Y_t = \underset{\bar{\P}\in\fP(\cG_{\smalltext{t}\smalltext{+}},\P)}{{\esssup}^\P} Y^{\bar{\P}}_t(T,\xi), \; \textnormal{$\P$--a.s.}, \; t \in [0,T], \; \P\in\fP.
	\]
	We note that for $t = T$, there is nothing to prove. So we fix $t \in [0,T)$, $\P\in\fP$, and then $\overline{\P}\in\fP(\cG_{t+},\P)$. Let $\tau \in\cT_{t,T}(\G_+)$ be arbitrary. By right-continuity of $\xi$ and $\overline{\P}$--uniform integrability of $\{\xi_{(\tau+1/n)\land T} \,|\, n \in \N\}$, we find
	\[
	\E^{\bar{\P}}\Big[\big|\E^{\bar{\P}}[\xi_\tau|\cG_{t+}] - \E^{\bar{\P}}[\xi_{(\tau+1/n)\land T}|\cG_{t+}]\big|\Big] 
	\leq \E^{\bar{\P}}\big[|\xi_\tau - \xi_{(\tau+1/n)\land T}|\big] \xrightarrow{n\rightarrow\infty} 0,
	\]
	and therefore, upon choosing a suitable subsequence (depending on $\overline{\P}$), if necessary,
	\[
	\E^{\bar{\P}}[\xi_\tau|\cG_{t+}] = \lim_{n \rightarrow\infty} \E^{\bar{\P}}[\xi_{(\tau+1/n)\land T}|\cG_{t+}], \; \textnormal{$\overline{\P}$--a.s.}
	\]
	Thus, for $n$ large enough such that $t + 1/n < T$, and since
	\[
		\E^{\bar{\P}}[Y^{\bar{\P}}_{(t+1/n)}(T,\xi)|\cG_{(t+1/n)}] = \E^{\bar{\P}}[Y^{\bar{\P}}_{(t+1/n)}(T,\xi)|\cF_{(t+1/n)}] \leq \cY_{t+1/n}, \; \textnormal{$\overline{\P}$--a.s.},
	\]
	we obtain
	\begin{align*}
		\E^{\bar{\P}}[\xi_\tau|\cG_{t+}] 
		&= \lim_{n\rightarrow\infty} \E^{\bar{\P}}[\xi_{(\tau+1/n)\land T}|\cG_{t+}] 
		= \lim_{n\rightarrow\infty} \E^{\bar{\P}}\big[\E^{\bar{\P}}[\xi_{(\tau+1/n)\land T}|\cG_{(t+1/n)+}]\big|\cG_{t+}\big]
		\leq \liminf_{n\rightarrow\infty} \E^{\bar{\P}}\big[Y^{\P}_{(t+1/n)}(T,\xi)\big|\cG_{t+}\big] \\
		&= \liminf_{n\rightarrow\infty} \E^{\bar{\P}}\big[\E^{\bar{\P}}[Y^{\bar{\P}}_{(t+1/n)}(T,\xi)|\cG_{(t+1/n)}]\big|\cG_{t+}\big] 
		\leq \liminf_{n\rightarrow\infty} \E^{\bar{\P}}\big[\cY_{t+1/n}\big|\cG_{t+}\big], \; \textnormal{$\overline{\P}$--a.s.}
	\end{align*}
	It follows from \eqref{eq_integrability_Y_0} that the family $\{\cY_{s} \,|\, s\in[0,T]\}$ is bounded in $\L^1(\overline{\P})$.
	Since $\cY_{t+1/n}$ converges to $Y_t$, $\overline{\P}$--a.s., it then follows from the backward super-martingale convergence theorem \cite[Theorem V.30]{dellacherie1982probabilities} that $\cY_{t+1/n}$ converges to $Y_t$ also in $\L^1(\overline{\P})$, and therefore
	\begin{equation}\label{eq_L1_convergence_Y}
		Y_t = \lim_{n\rightarrow\infty}\E^{\overline{\P}}[\cY_{t+1/n}|\cG_{t+}]
	\end{equation}
	in $\L^1(\overline{\P})$. It is straightforward to check that the sequence of conditional expectations on the right is $\overline{\P}$--a.s. monotonically non-decreasing in $n$ by the super-martingale property; therefore the limit also holds in the $\overline{\P}$--a.s. sense. We obtain
	\[
	\E^{\bar{\P}}[\xi_\tau|\cG_{t+}] \leq \liminf_{n\rightarrow\infty} \E^{\bar{\P}}\big[\cY_{t+1/n}\big|\cG_{t+}\big] = Y_t, \; \textnormal{$\overline{\P}$--a.s.},
	\]
	and then
	\[
	Y^{\overline{\P}}_t(T,\xi) 
	= \underset{\tau\in\cT_{\smalltext{t}\smalltext{,}\smalltext{T}}(\G_\smalltext{+})}{{\esssup}^{\bar\P}} \E^{\bar{\P}}[\xi_\tau|\cG_{t+}] \leq Y_t, \; \textnormal{$\overline{\P}$--a.s.}
	\]
	Since the terms above are $\cG_{t+}$-measurable, and since $\P = \overline{\P}$ on $\cG_{t+}$, the above line also holds outisde a $\P$--null set. By arbitraryness of $\overline{\P}\in\fP(\cG_{t+},\P)$, we obtain
	\[
	\underset{\bar{\P}\in\fP(\cG_{\smalltext{t}\smalltext{+}},\P)}{{\esssup}^\P} Y^{\bar{\P}}_t(T,\xi)
	= \underset{\bar{\P} \in \fP_\smalltext{0}(\cG_{t\smalltext{+}},\P)}{{\esssup}^\P} \underset{\tau\in\cT_{\smalltext{t}\smalltext{,}\smalltext{T}}(\G_\smalltext{+})}{{\esssup}^{\bar\P}} \E^{\bar{\P}}[\xi_\tau|\cG_{t+}]
	\leq Y_t , \; \textnormal{$\P$--a.s.}
	\]
	
	\medskip
	We turn to the converse inequality, which can be argued as follows. For each $n \in \N$ with $t+1/n < T$, we find $(\P^m_n)_{m \in \N} \subseteq \fP(\cF_{t+1/n},\P)$ such that
	\[
	\E^{\P^\smalltext{m}_\smalltext{n}} \big[ Y^{\P^\smalltext{m}_\smalltext{n}}_{t+1/n}(T,\xi) \big|\cF_{t+1/n}\big] \nearrow \cY_{t+1/n}, \; m\longrightarrow\infty, \; \textnormal{$\P$--a.s.}
	\] 
	Since $(\P^{m}_{n})_{m\in\N} \subseteq \fP(\cG_{t+},\P)$, we obtain
	\begin{align*}
		\E^\P[\cY_{t+1/n}|\cG_{t+}] 
		&= \lim_{m \rightarrow\infty} \E^{\P}\big[ \E^{\P^\smalltext{m}_\smalltext{n}} [ Y^{\P^\smalltext{m}_\smalltext{n}}_{t+1/n}(T,\xi) |\cF_{t+1/n}] \big|\cG_{t+}\big]
		= \lim_{m \rightarrow\infty} \E^{\P^\smalltext{m}_\smalltext{n}}\big[ \E^{\P^\smalltext{m}_\smalltext{n}} [ Y^{\P^\smalltext{m}_\smalltext{n}}_{t+1/n}(T,\xi) |\cF_{t+1/n}] \big|\cG_{t+}\big] \\
		&= \lim_{m \rightarrow\infty} \E^{\P^\smalltext{m}_\smalltext{n}}\big[ Y^{\P^\smalltext{m}_\smalltext{n}}_{t+1/n}(T,\xi) \big|\cG_{t+}\big] 
		\leq \liminf_{m \rightarrow \infty} Y^{\P^\smalltext{m}_\smalltext{n}}_{t}(T,\xi), \; \textnormal{$\P$--a.s.}, \; n \in \N,
	\end{align*}
	and then
	\[
	\E^\P[\cY_{t+1/n}|\cG_{t+}] 
	\leq \underset{\P^\smalltext{\prime}\in\fP(\cG_{\smalltext{t}\smalltext{+}},\P)}{{\esssup}^\P} Y^{\P^\smalltext{\prime}}_t(T,\xi), \; \textnormal{$\P$--a.s.}, \; n \in \N.
	\]
	With \eqref{eq_L1_convergence_Y} for $\P$ instead of $\overline{\P}$, we find
	\[
	Y_t = \lim_{n \rightarrow\infty} \E^\P[\cY_{t+1/n}|\cG_{t+}] 
	\leq \underset{\P^\smalltext{\prime}\in\fP(\cG_{\smalltext{t}\smalltext{+}},\P)}{{\esssup}^\P} Y^{\P^\smalltext{\prime}}_t(T,\xi)
	= \underset{\P^\smalltext{\prime} \in \fP_\smalltext{0}(\cG_{t\smalltext{+}},\P)}{{\esssup}^\P} \underset{\tau\in\cT_{\smalltext{t}\smalltext{,}\smalltext{T}}(\G_\smalltext{+})}{{\esssup}^{\P^\smalltext{\prime}}} \E^{\P^\smalltext{\prime}}[\xi_\tau|\cG_{t+}], \; \textnormal{$\P$--a.s.}
	\]
	
	\medskip
	{\bf Step 5:} Lastly, we establish \eqref{eq_Y_0_equaly_xi}. Since $\cF_0 = \{\varnothing,\Omega\}$, it follows from \cite[Equation VI.2.3]{dellacherie1982probabilities} that
	\[
	\E^\P[\xi_\tau] \leq \E^\P[Y_0] = \E^\P[Y_0|\cF_0] \leq \cY_0, \; \P\in\fP, \; \tau\in\cT_{0,T}(\G_+).
	\] 
	This yields
	\[
	\sup_{\P\in\fP} \sup_{\tau\in\cT_{\smalltext{0}\smalltext{,}\smalltext{T}}(\G_\smalltext{+})} \E^\P[\xi_\tau] \leq \sup_{\P\in\fP} \E^\P[Y_0] \leq \cY_0 \leq \sup_{\P\in\fP} \sup_{\tau\in\cT_{\smalltext{0}\smalltext{,}\smalltext{T}}(\G_\smalltext{+})} \E^\P[\xi_\tau],
	\]
	where the last inequality follows from arguments similar to those used in deriving \eqref{eq_integrability_Y_0}. This concludes the proof.
	
	\section{Conclusions}\label{sec_conclusions}
	
	The construction of the aggregator in \Cref{thm_construction_capital_process} depends on proving some sort of measurability of the Snell envelope with respect to the probability law. From the construction of the Snell envelope (more precisely, from the construction of the corresponding essential suprema), it is not clear whether the measurability of conditional expectations with respect to the probability law, as established in \cite[Lemma 3.1]{neufeld2014measurability}, carries over to the essential suprema thereof. Therefore, a natural approach is to approximate the Snell envelope by processes for which measurability in the probability law can be established. Here, we choose to view the Snell envelope as the first component of the solution to a reflected BSDE, and approximate it through a penalization scheme by BSDEs. Since the solutions of the corresponding BSDEs are constructed through a fixed-point argument, this leads to the necessary measurability condition of the Snell envelope. We rely on the approximation by penalization results from \cite{klimsiak2015reflected,klimsiak2013dirichlet}, and thus suppose, in particular, that $\xi$ is c\`adl\`ag and of class $(D)$. An additional advantage of considering c\`adl\`ag obstacles is that it allows us to regularise the super-martingale $\cY$ in the proof of \Cref{thm_construction_capital_process}, leading to the desired aggregator $Y$, which will then be a c\`adl\`ag super-martingale relative to $(\G_+,\fP)$. To derive the hedging duality, we subsequently apply the robust optional decomposition for c\`adl\`ag super-martingales.
	
	\medskip
	Given the works on optimal stopping problems or reflected BSDEs with general optional processes \cite{grigorova2017reflected,grigorova2020optimal,kobylanski2012optimal,possamai2024reflections}, it is natural to wonder whether the results presented here can be extended to American options $\xi$ without any regularity assumptions on the paths $t \longmapsto \xi_t(\omega)$. First, the approach outlined above to construct the aggregator of the Snell envelope would need to be modified or extended. We do not exclude the possibility of a different method to prove the measurability of the Snell envelope with respect to the probability law, which could eliminate the need for the penalization approach, particularly as it does not appear to be well suited for irregular obstacles. Secondly, if this leads to an aggregator of Snell envelopes relative to $\G$ instead of $\G_+$, that is only a l\`adl\`ag $(\G,\fP)$--super-martingale, we would then require an optional decomposition for l\`adl\`ag $(\G,\fP)$--super-martingales. We emphasize that for l\`adl\`ag $(\G_+,\fP)$--super-martingales $Y$, the optional decomposition is an application of the following trick; we refer to \cite[p.~51~ff.]{mertens1972theorie} for details. If $Y$ is a l\`adl\`ag $(\G^\P_+,\P)$--super-martingale, the process
	\[
	Y^\prime = Y + \sum_{0 \leq s < \cdot} (Y_s - Y_{s+})
	\]
	is a right-continuous $(\G^\P_+,\P)$--super-martingale. Here, one uses the fact that $Y_{s} \geq Y_{s+}$, $\P$--a.s.; see \cite[Theorem~VI.2]{dellacherie1982probabilities}. Applying the (robust) optional decomposition to $Y^\prime$ then yields an $S$-integrable process $Z$ such that $Y^\prime - Y^\prime_0 - Z \bcdot S$ is $\P$--a.s. non-increasing. This, in turn, implies that
	\[
	Y - Y_0 - Z \bcdot S
	\]
	is also $\P$--a.s. non-increasing. However, this argument crucially relies on $Y$ being a $(\G^\P_+,\P)$--super-martingale. The same reasoning does not apply if $Y$ is only a $(\G^\P,\P)$--super-martingale, since in that case we only know, in general, that $Y_s \geq \E^\P[Y_{s+}|\cG^\P_s]$. Thus, if the aggregator of the Snell envelopes were to be a l\`adl\`ag $(\G,\fP)$--super-martingale, one would need an optional decomposition that applies to l\`adl\`ag super-martingales and filtrations which are not right-continuous.
	
	\medskip
	Another direction for further generalization is to consider random time horizons that may be unbounded. The construction of the aggregator suggests that one could also consider obstacles $\xi$ defined on such horizons, provided that the Snell envelopes can be approximated by BSDEs for which measurability with respect to the probability parameter can be appropriately established. We again do not exclude the possibility that these measurability issues can be overcome with other tools.
	
	\medskip
	In summary, verifying measurability in the context of irregular American options, extending the optional decomposition to l\`adl\`ag super-martingales without right-continuous filtrations, or considering options on random and unbounded horizons would require careful justification, provided such extensions are even possible. We therefore leave this for future research.
	
{\small\bibliography{bibliography}}
	
\end{document}